\documentclass[review]{elsarticle}

\usepackage{lineno,hyperref}
\usepackage{amsmath}
\usepackage{amssymb}
\usepackage{algorithmic}
\usepackage{algorithm}
\usepackage{caption}
\usepackage{mathtools}
\usepackage{changes}
\usepackage{multirow}
\usepackage{verbatim}
\usepackage{mathrsfs}
\usepackage{graphicx}
\usepackage{subcaption}
\usepackage{amsmath,amsthm,bm,mathrsfs}

\theoremstyle{plain}
\newtheorem{theorem}{Theorem}[section]

\newtheorem{corollary}[theorem]{Corollary}
\newtheorem{proposition}[theorem]{Proposition}

\theoremstyle{definition}
\newtheorem{definition}[theorem]{Definition}

\theoremstyle{remark}
\newtheorem{remark}{Remark}

\modulolinenumbers[5]

\journal{ArXiv.org}

\begin{document}

\begin{frontmatter}

\title{$\mathcal{H}_2$-optimal Model Order Reduction of Linear Quadratic Output Systems in Finite Frequency Range}

\author[uz]{Umair~Zulfiqar}
\author[qs]{Qiu-Yan~Song\corref{mycorrespondingauthor}}
\cortext[mycorrespondingauthor]{Corresponding author}
\ead{qysong@shu.edu.cn}
\author[zx]{Zhi-Hua~Xiao}
\author[mud]{Mohammad~Monir~Uddin}
\author[vs]{Victor~Sreeram}
\address[uz]{School of Electronic Information and Electrical Engineering, Yangtze University, Jingzhou, Hubei, 434023, China}
\address[qs]{School of Mechatronic Engineering and Automation, Shanghai University, Shanghai, 200444, China}
\address[zx]{School of Statistics and Data Science, Nanjing Audit University, Nanjing, Jiangsu, 211815, China}
\address[mud]{Department of Mathematics and Physics, North South University, Dhaka, 1229, Bangladesh}
\address[vs]{Department of Electrical, Electronic, and Computer Engineering, The University of Western Australia, Perth, 6009, Australia}
\begin{abstract}
Linear quadratic output systems constitute an important class of dynamical systems with numerous practical applications. When the order of these models is exceptionally high, simulating and analyzing these systems becomes computationally prohibitive. In such instances, model order reduction offers an effective solution by approximating the original high-order system with a reduced-order model while preserving the system's essential characteristics.

In frequency-limited model order reduction, the objective is to maintain the frequency response of the original system within a specified frequency range in the reduced-order model. In this paper, a mathematical expression for the frequency-limited $\mathcal{H}_2$ norm is derived, which quantifies the error within the desired frequency interval. Subsequently, the necessary conditions for a local optimum of the frequency-limited $\mathcal{H}_2$ norm of the error are derived. The inherent difficulty in satisfying these conditions within a Petrov-Galerkin projection framework is also discussed. Using the optimality conditions and the Petrov-Galerkin projection, a stationary point iteration algorithm is proposed, which approximately satisfies these optimality conditions upon convergence. The main computational effort in the proposed algorithm involves solving sparse-dense Sylvester equations. These equations are frequently encountered in \(\mathcal{H}_2\) model order reduction algorithms and can be solved efficiently. Moreover, the algorithm bypasses the requirement of matrix logarithm computation, which is typically necessary for most frequency-limited reduction methods and can be computationally demanding for high-order systems. An illustrative example is provided to numerically validate the developed theory. The proposed algorithm's effectiveness in accurately approximating the original high-order model within the specified frequency range is demonstrated through the reduction of an advection-diffusion equation-based model, commonly used in model reduction literature for testing algorithms. Additionally, the algorithm's computational efficiency is highlighted by successfully reducing a flexible space structure model of order one million.
\end{abstract}

\begin{keyword}
$\mathcal{H}_2$-optimal\sep frequency-limited\sep model order reduction\sep projection\sep reduced-order model\sep quadratic output
\end{keyword}

\end{frontmatter}

\section{Introduction}
This study focuses on a specific class of nonlinear dynamical systems with weak nonlinearity. These systems have linear time-invariant (LTI) state equations but feature quadratic nonlinear terms in the output equation, referred to as linear quadratic output (LQO) systems \cite{montenbruck2017linear}. They naturally arise in scenarios where it is necessary to observe quantities involving products of state components in either the time or frequency domain. For example, they are used in applications that quantify energy or power, such as the internal energy functional of a system \cite{van2006port} or the objective cost function in optimal quadratic control problems \cite{picinbono1988optimal}. Additionally, they are used to measure the deviation of a state's coordinates from a reference point, such as the root mean squared displacement of spatial coordinates around an excitation point, or in stochastic modeling for calculating the variance of a random variable \cite{aumann2023structured}. These models are also used in the mathematical modeling of mechanical systems, including mass-spring-damper systems \cite{depken1974observability}, random vibration analysis \cite{haasdonk2013reduced,lutes2004random}, and electrical circuits governed by time-harmonic Maxwell's equations \cite{hammerschmidt2015reduced,hess2016output}.

Consider a LQO system described by the following state and output equations:
\begin{align}
 G:=
  \begin{cases}
   \dot{x}(t)=Ax(t)+Bu(t),\\
  y(t)=Cx(t)+\begin{bmatrix}x(t)^TM_1x(t)\\\vdots\\x(t)^TM_px(t)\end{bmatrix},
  \end{cases}\label{steq:1}
\end{align} wherein $x(t)\in\mathbb{R}^{n\times 1}$ represents the state vector, $u(t)\in\mathbb{R}^{m\times 1}$ represents inputs, and $y(t)\in\mathbb{R}^{p\times 1}$ represents outputs. $(A,B,C,M_1,\cdots,M_p)$ is the state-space realization of $G$ with $A\in\mathbb{R}^{n\times n}$, $B\in\mathbb{R}^{n\times m}$, $C\in\mathbb{R}^{p\times n}$, and $M_i\in\mathbb{R}^{n\times n}$. The state equation in (\ref{steq:1}) is identical to that of standard LTI systems. However, the output equation in (\ref{steq:1}) introduces a nonlinearity in the form of the quadratic function of states $x(t)^TM_ix(t)$. The matrix \( M_i \) is assumed to be symmetric without any loss of generality, as a symmetric \( M_i^\prime \) can always be constructed as \( M_i^\prime =\frac{M_i + M_i^T}{2} \), preserving the quadratic relation \( x(t)^T M_i^\prime x(t) \).

Let us denote $y_1(t)=Cx(t)$ and $y_{2,i}=x(t)^TM_ix(t)$. The input-output mapping between
$u(t)$ and $y_1(t)$ is represented by the following transfer function:
\begin{align}
G_1(s)&=C(sI-A)^{-1}B.\nonumber
\end{align} Additionally, the input-output mapping between $u(t)$ and $y_{2,i}(t)$ is represented by the following multivariate transfer function:
\begin{align}
G_{2,i}(s_1,s_2)&=B^T(s_1I-A)^{-*}M_i(s_2I-A)^{-1}B,\nonumber
\end{align}cf. \cite{reiter2024interpolatory}.

To ensure high fidelity in the mathematical modeling of complex physical phenomena, it is often necessary to use dynamical systems with very high orders, sometimes exceeding several thousand. This high order $n$ makes it computationally difficult or even prohibitive to simulate and analyze the model (\ref{steq:1}). Therefore, it is important to approximate (\ref{steq:1}) with a reduced-order model (ROM) of much lower order $k$ (where $k\ll n$). This approach simplifies the simulation and analysis processes. The process of constructing such a ROM while preserving the key features of the original model is known as model order reduction (MOR); refer to \cite{van2012model,benner2005dimension,antoulas2004model,antoulas2015model,antoulas2020interpolatory} for a deeper understanding of the topic.

Let us denote the $k^{th}$-order ROM of $G$ as $G_k$, characterized by the following state and output equations:
\begin{align}
 G_k:=
  \begin{cases}
   \dot{x_k}(t)=A_kx_k(t)+B_ku(t),\\
  y_k(t)=C_kx_k(t)+\begin{bmatrix}x_k(t)^TM_{k,1}x_k(t)\\\vdots\\x_k(t)^TM_{k,p}x_k(t)\end{bmatrix},
  \end{cases}\label{eq:2}
\end{align} wherein
$A_k=W_k^TAV_k\in\mathbb{R}^{k\times k}$, $B_k=W_k^TB\in\mathbb{R}^{k\times m}$, $C_k=CV_k\in\mathbb{R}^{p\times k}$, and $M_{k,i}=V_k^TM_iV_k\in\mathbb{R}^{k\times k}$, satisfying the Petrov-Galerkin projection condition $W_k^TV_k=I$. The projection matrices $V_k\in\mathbb{R}^{n\times k}$ and $W_k\in\mathbb{R}^{n\times k}$ project $G$ onto a reduced subspace to obtain the ROM $G_k$. Various MOR methods differ in how they construct $V_k$ and $W_k$. The choice of $V_k$ and $W_k$ depends on the specific characteristics of $G$ that need to be preserved in $G_k$.

Let $y_{k,1}(t)=C_kx_k(t)$ and $y_{k,2,i}=x_k(t)^TM_{k,i}x_k(t)$. The input-output relationships between $u(t)$ and $y_{k,1}(t)$, as well as $u(t)$ and $y_{k,2,i}(t)$, are described by the following transfer functions:
\begin{align}
G_{k,1}(s)&=C_k(sI-A_k)^{-1}B_k,\nonumber\\
G_{k,2,i}(s_1,s_2)&=B_k^T(s_1I-A_k)^{-*}M_{k,i}(s_2I-A_k)^{-1}B_k.\nonumber
\end{align}Throughout this paper, it is assumed that both $A$ and $A_k$ are Hurwitz.

The Balanced Truncation (BT) method, introduced in 1981, is a prominent technique for MOR \cite{moore1981principal}. This approach retains states that significantly effect energy transfer between inputs and outputs, discarding those with minimal impact as indicated by their Hankel singular values. A notable advantage of BT is its ability to estimate errors in advance of creating the ROM \cite{enns1984model}. Moreover, BT preserves the stability of the original system. Originally proposed for LTI systems, BT's application has broadened to include descriptor systems \cite{mehrmann2005balanced}, second-order systems \cite{reis2008balanced}, linear time-varying systems \cite{sandberg2004balanced}, parametric systems \cite{son2021balanced}, nonlinear systems \cite{kramer2022balanced}, and bilinear systems \cite{benner2017truncated}. Additionally, BT has been tailored to maintain specific system properties such as positive realness \cite{wong2007fast}, bounded realness \cite{guiver2013bounded}, passivity \cite{wong2004passivity}, and special structural characteristics \cite{sarkar2023structure}. For a detailed overview of the various BT algorithms, refer to the survey \cite{gugercin2004survey}. BT has been adapted for LQO systems in \cite{van2010model,pulch2019balanced,benner2021gramians}. Of these algorithms, only the one proposed in \cite{benner2021gramians} preserves the LQO structure in the ROM.

The $\mathcal{H}_2$-optimal MOR problem for standard LTI systems has been thoroughly explored in the literature. Wilson's conditions, which are necessary conditions for achieving a local optimum of the $\mathcal{H}_2$ norm of error, are outlined in \cite{wilson1970optimum}. The iterative rational Krylov algorithm (IRKA) was the first to apply interpolation theory to meet these conditions \cite{gugercin2008h_2}. Other algorithms using Sylvester equations and projection have been developed and enhanced for improved robustness in \cite{xu2011optimal} and \cite{MPIMD11-11}. Recently, the $\mathcal{H}_2$-optimal MOR problem for LQO systems has been addressed in \cite{reiter2024h2}, where a Sylvester equation-based algorithm has been proposed to achieve a local optimum upon convergence.

Many MOR problems are inherently frequency-limited, with certain frequency ranges being more important. For example, when creating a ROM for a notch filter, it is crucial to minimize the approximation error near the notch frequency \cite{aldhaheri2006frequency}. Similarly, for closed-loop stability, the ROM of a plant must accurately capture the system’s behavior in the crossover frequency region \cite{obinata2012model,ennsphd}. In interconnected power systems, low-frequency oscillations are essential for small-signal stability studies, so the ROM should accurately represent the behavior within the frequency range of inter-area and inter-plant oscillations \cite{zulfiqar2019finite,zulfiqar2020frequency}. This need has led to frequency-limited MOR, which focuses on achieving high accuracy within specific frequency intervals rather than the entire spectrum \cite{gawronski1990model}.

The frequency-limited MOR problem aims to create a ROM that ensures that
\begin{align}
||G_1(j\nu)-G_{k,1}(j\nu)||\hspace*{0.5cm}\textnormal{and}\hspace*{0.5cm}||G_{2,i}(j\nu_1,j\nu_2)-G_{k,2,i}(j\nu_1,j\nu_2)||\nonumber
\end{align} are small when $\nu$, $\nu_1$, and $\nu_2$ are within the specified frequency range of $[0,\omega]$ rad/sec.

BT typically offers an accurate approximation of the original model across the entire frequency spectrum. In \cite{gawronski1990model}, BT was adapted to address the frequency-limited MOR problem, resulting in the frequency-limited BT (FLBT) algorithm. However, FLBT does not preserve the stability and \textit{a priori} error bounds that BT does. The computational aspects of FLBT and efficient methods for handling large-scale systems are discussed in \cite{benner2016frequency}. Additionally, FLBT has been expanded to cover a wider range of systems, including descriptor systems \cite{imran2015model}, second-order systems \cite{benner2021frequency}, and bilinear systems \cite{shaker2013frequency}. FLBT has been recently generalized for LQO systems in \cite{song2024balanced}.

The frequency-limited $\mathcal{H}_2$ norm for LTI systems is defined in \cite{petersson2014model}, with Gramian-based conditions outlined for achieving a local optimum. Generalizations of the iterative rational Krylov algorithm (IRKA) have resulted in the development of the frequency-limited IRKA (FLIRKA) in \cite{vuillemin2013h2,du2021computational,zulfiqar2023frequency}, which approximately fulfills these conditions. The computation of Gramians in BT and FLBT is inherently expensive, necessitating low-rank approximations for their practical use in large-scale systems. For LQO systems, computing Gramians in FLBT is significantly more demanding than in LTI systems. As a result, various low-rank algorithms for Gramian approximation have been investigated in \cite{song2024balanced}. In $\mathcal{H}_2$-optimal MOR, the main computational burden lies in solving sparse-dense Sylvester equations, which is significantly more efficient than Gramian computations in FLBT. This computational efficiency advantage motivates the investigation of frequency-limited $\mathcal{H}_2$-optimal MOR for LQO systems, which is addressed in this paper.

This research work makes several important contributions. First, it introduces the frequency-limited $\mathcal{H}_2$ norm ($\mathcal{H}_{2,\omega}$ norm) for LQO systems and shows how to compute it using frequency-limited system Gramians defined in \cite{song2024balanced}. Second, it derives the necessary conditions for achieving a local optimum of $||G - G_k||_{\mathcal{H}_{2,\omega}}^2$. Third, it compares these conditions with those for standard $\mathcal{H}_2$-optimal MOR \cite{reiter2024h2}, highlighting that Petrov-Galerkin projection generally cannot achieve a local optimum in the frequency-limited scenario. Fourth, a stationary point algorithm based on Petrov-Galerkin projection is proposed, which approximately satisfies the necessary optimality conditions upon convergence. Fifth, an efficient method is introduced to approximate the computationally intensive matrix logarithm that appears in most frequency-limited MOR algorithms. The paper includes numerical examples demonstrating the algorithm's accuracy within the specified frequency range and showing its superiority over existing methods. The computational efficiency of the proposed algorithm is also illustrated by reducing a test system of the order of one million.
\section{Literature Review}
In this section, we will briefly explore two key MOR algorithms relevant to LQO systems within the context of the problem at hand. The first is the FLBT \cite{song2024balanced}, and the second is the $\mathcal{H}_2$-optimal MOR method \cite{reiter2024h2}.
\subsection{Frequency-limited Balanced Truncation (FLBT) \cite{song2024balanced}}
FLBT constructs the ROM by truncating states that contribute minimally to the input-output energy transfer within the desired frequency range of $[0,\omega]$ rad/sec. This is achieved by constructing a frequency-limited balanced realization using frequency-limited Gramians and then truncating the states corresponding to the smallest frequency-limited Hankel singular values.

The frequency-limited controllability Gramian $P_\omega$ within the desired frequency interval $[0,\omega]$ rad/sec is given by
\begin{align}
P_\omega=\frac{1}{2\pi}\int_{-\omega}^{\omega}(j\nu I-A)^{-1}BB^T(j\nu I -A)^{-*}d\nu.
\end{align}
In the majority of frequency-limited MOR algorithms, the following integral expression arises, which has been shown to be equivalent to a matrix-valued logarithm in \cite{gawronski1990model,petersson2013nonlinear}. This integral is defined as follows:
\begin{align}
F_\omega=\frac{1}{2\pi}\int_{-\omega}^{\omega}(j\nu I-A)^{-1}d\nu=\frac{j}{\pi}\ln(-j\omega I-A).
\end{align}Several algorithms have been investigated in \cite{higham2008functions} for computing matrix-valued logarithms of this form; however, their computational complexity makes them impractical for large-scale systems. With this, $P_\omega$ can be computed by solving the following Lyapunov equation:
\begin{align}
AP_\omega+P_\omega A^T+ F_\omega BB^T+BB^TF_\omega^*&=0.\label{eq:9}
\end{align}
The frequency-limited observability Gramian $Q_\omega = Y_\omega + Z_\omega$ within the frequency range $[0,\omega]$ rad/sec is defined as
\begin{align}
Y_\omega&=\frac{1}{2\pi}\int_{-\omega}^{\omega}(j\nu I-A)^{-*}C^TC(j\nu I -A)^{-1}d\nu,\\
Z_\omega&=\frac{1}{2\pi}\int_{-\omega}^{\omega}(j\nu_1 I-A)^{-*}\Bigg(\sum_{i=1}^{p}M_i\Big(\frac{1}{2\pi}\int_{-\omega}^{\omega}(j\nu_2 I-A)^{-1}BB^T\nonumber\\
&\hspace*{3.8cm}(j\nu_2 I-A)^{-*}d\nu_2\Big)M_i\Bigg)(j\nu_1 I-A)^{-1}d\nu_1\nonumber\\
&=\frac{1}{2\pi}\int_{-\omega}^{\omega}(j\nu_1 I-A)^{-*}\big(\sum_{i=1}^{p}M_iP_\omega M_i\big)(j\nu_1 I-A)^{-1}d\nu_1.
\end{align} $Y_\omega$, $Z_\omega$, and $Q_\omega$ can be calculated by solving the following Lyapunov equations:
\begin{align}
A^TY_\omega+Y_\omega A+F_\omega^*C^TC+C^TCF_\omega&=0,\\
A^TZ_\omega+Z_\omega A+\sum_{i=1}^{p}\big(F_\omega^*M_iP_\omega M_i+M_iP_\omega M_iF_\omega\big)&=0,\\
A^TQ_\omega+Q_\omega A+F_\omega^*C^TC+C^TCF_\omega\hspace*{2cm}&\nonumber\\
+\sum_{i=1}^{p}\big(F_\omega^*M_iP_\omega M_i+M_iP_\omega M_iF_\omega\big)&=0.\label{new_eq:10}
\end{align}
The frequency-limited Hankel singular values $\sigma_i$ are defined as
\begin{align}
\sigma_i=\sqrt{\lambda_i(P_\omega Q_\omega)}\hspace*{0.5cm}\textnormal{for} \hspace*{0.5cm}i=1,\cdots,n,\nonumber
\end{align}where $\lambda_i(\cdot)$ denotes the eigenvalues. The projection matrices in FLBT are then computed such that $W_k^T P_\omega W_k = V_k^T Q_\omega V_k = \text{diag}(\sigma_1, \cdots, \sigma_k)$, where $\sigma_1, \cdots, \sigma_k$ are the $k$ largest frequency-limited Hankel singular values of $G$. The ROM $G_k$ obtained using these projection matrices provides high accuracy within the specified frequency range.
\subsection{$\mathcal{H}_2$-optimal MOR Algorithm (HOMORA) \cite{reiter2024h2}}
HOMORA is an efficient algorithm for LQO systems that, like BT, ensures good accuracy across the entire frequency spectrum rather than targeting a specific frequency range. Unlike BT, however, it eliminates the need to compute Lyapunov equations, which are computationally expensive in large-scale settings.

Let us define the matrices $P_{12}$, $P_k$, $Y_{12}$, $Y_k$, $Z_{12}$, $Z_k$, $Q_{12}$, and $Q_k$, which satisfy the following set of linear matrix equations:
\begin{align}
AP_{12}+P_{12}A_k^T+BB_k^T&=0,\nonumber\\
A_kP_k+P_kA_k^T+B_kB_k^T&=0,\nonumber\\
A^TY_{12}+Y_{12}A_k+C^TC_k&=0,\nonumber\\
A_k^TY_k+Y_kA_k+C_k^TC_k&=0,\nonumber\\
A^TZ_{12}+Z_{12}A_k+\sum_{i=1}^{p}M_iP_{12}M_{k,i}&=0,\nonumber\\
A_k^TZ_k+Z_kA_k+\sum_{i=1}^{p}M_{k,i}P_kM_{k,i}&=0,\nonumber\\
A^TQ_{12}+Q_{12}A_k+C^TC_k+\sum_{i=1}^{p}M_iP_{12}M_{k,i}&=0,\nonumber\\
A_k^TQ_k+Q_kA_k+C_k^TC_k+\sum_{i=1}^{p}M_{k,i}P_kM_{k,i}&=0.\nonumber
\end{align}
According to \cite{reiter2024h2}, the necessary conditions for achieving a local optimum of the (squared) $\mathcal{H}_2$-norm of the error, denoted as $||G - G_k||_{\mathcal{H}_2}^2$, are described by the following set of equations:
\begin{align}
-(Y_{12}+2Z_{12})^TP_{12}+(Y_k+2Z_k)P_k&=0,\label{op01}\\
-P_{12}^TM_iP_{12}+P_kM_{k,i}P_k&=0,\label{op02}\\
-(Y_{12}+2Z_{12})^TB+(Y_k+2Z_k)B_k&=0,\label{op03}\\
-CP_{12}+C_kP_k&=0.\label{op04}
\end{align}
Furthermore, it is shown that these optimality conditions can be met by setting the projection matrices as $V_k = P_{12}$ and $W_k = (Y_{12} + 2Z_{12})\big(P_{12}^T(Y_{12} + 2Z_{12})\big)^{-1}$. Starting with an initial guess for the ROM, the projection matrices are iteratively updated until convergence, at which point the optimality conditions (\ref{op01})-(\ref{op04}) are satisfied.
\section{Main Work}
The computation of Lyapunov equations (\ref{eq:9}) and (\ref{new_eq:10}) is costly in large-scale settings. In the infinite-frequency case, HOMORA holds an advantage over BT in that it avoids similar Lyapunov equations. Inspired by HOMORA's computational efficiency compared to BT, we now explore its frequency-limited counterpart, which likewise avoids the computation of Lyapunov equations (\ref{eq:9}) and (\ref{new_eq:10}).

In this section, we define the frequency-limited $\mathcal{H}_2$ norm for LQO systems and establish its connection to the frequency-limited observability Gramian $Q_\omega$. We then derive the necessary conditions for achieving a local optimum of the (squared) frequency-limited $\mathcal{H}_2$ norm of the error. Building on these optimality conditions, we introduce a projection-based iterative algorithm that closely satisfies these conditions upon convergence. The inherent limitation of exactly satisfying these optimality conditions within the projection framework is also addressed. Finally, the computational aspects of the proposed algorithm are briefly discussed.
\subsection{$\mathcal{H}_{2,\omega}$ norm Definition}
The classical $\mathcal{H}_2$ norm for LQO systems is defined in the frequency domain as follows:
\begin{align}
||G||_{\mathcal{H}_2}&=\Bigg[\operatorname{trace}\Big(\frac{1}{2\pi}\int_{-\infty}^{\infty}G_1^*(j\nu)G_1(j\nu)d\nu\nonumber\\
&\hspace*{1cm}+\frac{1}{(2\pi)^2}\int_{-\infty}^{\infty}\int_{-\infty}^{\infty}\sum_{i=1}^{p}G_{2,i}^*(j\nu_1,j\nu_2)G_{2,i}(j\nu_1,j\nu_2)d\nu_1d\nu_2\Big)\Bigg]^{-\frac{1}{2}},\nonumber
\end{align}cf. \cite{reiter2024interpolatory,reiter2024h2}.
The $\mathcal{H}_2$ norm quantifies the output response's power to unit white noise across the entire frequency spectrum. However, for the problem at hand, we are only interested in the output response's power within a specific, limited frequency range. This leads to the definition of the frequency-limited $\mathcal{H}_2$ norm.
\begin{definition}
The frequency-limited $\mathcal{H}_2$ norm of the LQO system within the frequency interval $[0,\omega]$ rad/sec is defined as
\begin{align}
||G||_{\mathcal{H}_{2,\omega}}&=\Big[\operatorname{trace}\Big(\frac{1}{2\pi}\int_{-\omega}^{\omega}G_1^*(j\nu)G_1(j\nu)d\nu\nonumber\\
&\hspace*{1cm}+\frac{1}{(2\pi)^2}\int_{-\omega}^{\omega}\int_{-\omega}^{\omega}\sum_{i=1}^{p}G_{2,i}^*(j\nu_1,j\nu_2)G_{2,i}(j\nu_1,j\nu_2)d\nu_1d\nu_2\Big)\Big]^{-\frac{1}{2}}.\nonumber
\end{align}
\end{definition}
\begin{proposition}
The $\mathcal{H}_{2,\omega}$ norm is related to the frequency-limited observability Gramian $Q_\omega$ as follows:
\begin{align}
||G||_{\mathcal{H}_{2,\omega}}=\sqrt{\operatorname{trace}(B^TQ_\omega B)}.\nonumber
\end{align}
\end{proposition}
\begin{proof}
Observe that
\begin{align}
&\operatorname{trace}\Big(\frac{1}{2\pi}\int_{-\omega}^{\omega}G_1^*(j\nu)G_1(j\nu)d\nu\Big)\nonumber\\
&=\operatorname{trace}\Big(B^T\Big[\frac{1}{2\pi}\int_{-\omega}^{\omega}(j\nu I-A)^{-*}C^TC(j\nu I-A)^{-1}\Big]Bd\nu\Big)\nonumber\\
&=\operatorname{trace}(B^TY_\omega B).\nonumber
\end{align}
Additionally, note that
\begin{align}
&\operatorname{trace}\Big(\frac{1}{(2\pi)^2}\int_{-\omega}^{\omega}\int_{-\omega}^{\omega}\sum_{i=1}^{p}G_{2,i}^*(j\nu_1,j\nu_2)G_{2,i}(j\nu_1,j\nu_2)d\nu_1d\nu_2\Big)\nonumber\\
&=\operatorname{trace}\Bigg(B^T\Big[\frac{1}{2\pi}\int_{-\omega}^{\omega}(j\nu_2 I-A)^{-*}\Big(\sum_{i=1}^{p}M_i\Big(\frac{1}{2\pi}\int_{-\omega}^{\omega}(j\nu_1 I-A)^{-1}BB^T\nonumber\\
&\hspace*{4cm}(j\nu_1 I-A)^{-*}d\nu_1\Big)M_i\Big)(j\nu_2 I-A)^{-1}d\nu_2\Big]B\Bigg)\nonumber\\
&=\operatorname{trace}(B^TZ_\omega B).\nonumber
\end{align}
Therefore, we have $||G||_{\mathcal{H}_{2,\omega}}=\sqrt{\operatorname{trace}\big(B^T(Y_\omega+Z_\omega)B\big)}=\sqrt{\operatorname{trace}\big(B^T(Q_\omega)B\big)}$.
\end{proof}
\subsection{$\mathcal{H}_{2,\omega}$ Norm of the Error}
Let us define $G-G_k$ with the following state-space equations
\begin{align}
 G-G_k:=
  \begin{cases}
   \dot{x_e}(t)=\begin{bmatrix}x(t)\\x_k(t)\end{bmatrix}=A_ex_e(t)+B_eu(t),\\
  y_e(t)=y(t)-y_k(t)=C_ex_e(t)+\begin{bmatrix}x_e(t)^TM_{e,1}x_e(t)\\\vdots\\x_e(t)^TM_{e,p}x_e(t)\end{bmatrix},\nonumber
  \end{cases}
\end{align}wherein
\begin{align}
A_e&=\begin{bmatrix}A&0\\0&A_k\end{bmatrix},&B_e&=\begin{bmatrix}B\\B_k\end{bmatrix},\nonumber\\
M_{e,i}&=\begin{bmatrix}M_i&0\\0&-M_{k,i}\end{bmatrix},&C_e&=\begin{bmatrix}C&-C_k\end{bmatrix}.\label{partreal}
\end{align}
Let us define $F_{e,\omega}$ as follows
\begin{align}
F_{e,\omega}=\frac{1}{2\pi}\int_{-\omega}^{\omega}(j\nu I-A_e)^{-1}d\nu=\frac{j}{\pi}ln(-j\omega I-A_e).\nonumber
\end{align}
Then the frequency-limited controllability Gramian $P_{e,\omega}$ and the frequency-limited observability Gramian $Q_{e,\omega}=Y_{e,\omega}+Z_{e,\omega}$ of realization $(A_e$,$B_e$,$C_e$,$M_{e,1}$, $\cdots$,$M_{e,p})$ can be determined by solving the following Lyapunov equations:
\begin{align}
&\hspace*{3cm}A_eP_{e,\omega}+P_{e,\omega} A_e^T+ F_{e,\omega} B_eB_e^T+B_eB_e^TF_{e,\omega}^*=0,\nonumber\\
&\hspace*{3.1cm}A_e^TY_{e,\omega}+Y_{e,\omega} A_e+F_{e,\omega}^*C_e^TC_e+C_e^TC_eF_{e,\omega}=0,\nonumber\\
&A_e^TZ_{e,\omega}+Z_{e,\omega} A_e+\sum_{i=1}^{p}\big(F_{e,\omega}^*M_{e,i}P_{e,\omega} M_{e,i}+M_{e,i}P_{e,\omega} M_{e,i}F_{e,\omega}\big)=0,\nonumber\\
&A_e^TQ_{e,\omega}+Q_{e,\omega} A_e+F_{e,\omega}^*C_e^TC_e+C_e^TC_eF_{e,\omega}\nonumber\\
&\hspace*{2.75cm}+\sum_{i=1}^{p}\big(F_{e,\omega}^*M_{e,i}P_{e,\omega}M_{e,i}+M_{e,i}P_{e,\omega} M_{e,i}F_{e,\omega}\big)=0.\nonumber
\end{align}
Let us partition $P_{e,\omega}$, $Y_{e,\omega}$, $Z_{e,\omega}$, and $Q_{e,\omega}$ according to (\ref{partreal}) as follows:
\begin{align}
P_{e,\omega}&=\begin{bmatrix}P_\omega&P_{12,\omega}\\P_{12,\omega}^*&P_{k,\omega}\end{bmatrix},&
Y_{e,\omega}&=\begin{bmatrix}Y_\omega&-Y_{12,\omega}\\-Y_{12,\omega}^*&Y_{k,\omega}\end{bmatrix},\nonumber\\
Z_{e,\omega}&=\begin{bmatrix}Z_\omega&-Z_{12,\omega}\\-Z_{12,\omega}^*&Z_{k,\omega}\end{bmatrix},&
Q_{e,\omega}&=\begin{bmatrix}Q_\omega&-Q_{12,\omega}\\-Q_{12,\omega}^*&Q_{k,\omega}\end{bmatrix}.\nonumber
\end{align}
Additionally, define $F_{k,\omega}$ as
\begin{align}
F_{k,\omega}=\frac{1}{2\pi}\int_{-\omega}^{\omega}(j\nu I-A_k)^{-1}d\nu=\frac{j}{\pi}ln(-j\omega I-A_k).\nonumber
\end{align}
The following linear matrix equations then hold:
\begin{align}
&\hspace*{3.9cm}AP_{12,\omega}+P_{12,\omega}A_k^T+F_\omega BB_k^T+BB_k^TF_{k,\omega}^*=0,\label{eq:24}\\
&\hspace*{3.45cm}A_kP_{k,\omega}+P_{k,\omega}A_k^T+F_{k,\omega}B_kB_k^T+B_kB_k^TF_{k,\omega}^*=0,\label{eq:25}\\
&\hspace*{3.54cm}A^TY_{12,\omega}+Y_{12,\omega}A_k+F_\omega^*C^TC_k+C^TC_kF_{k,\omega}=0,\label{eq:26}\\
&\hspace*{3.55cm}A_k^TY_{k,\omega}+Y_{k,\omega}A_k+F_{k,\omega}^*C_k^TC_k+C_k^TC_kF_{k,\omega}=0,\label{eq:27}\\
&\hspace*{0.52cm}A^TZ_{12,\omega}+Z_{12,\omega}A_k+\sum_{i=1}^{p}\big(F_\omega^*M_iP_{12,\omega}M_{k,i}+M_iP_{12,\omega}M_{k,i}F_{k,\omega}\big)=0,\label{eq:28}\\
&\hspace*{0.3cm}A_k^TZ_{k,\omega}+Z_{k,\omega}A_k+\sum_{i=1}^{p}\big(F_{k,\omega}^*M_{k,i}P_{k,\omega}M_{k,i}+M_{k,i}P_{k,\omega}M_{k,i}F_{k,\omega}\big)=0,\label{eq:29}\\
&A^TQ_{12,\omega}+Q_{12,\omega}A_k+F_\omega^*C^TC_k+C^TC_kF_{k,\omega}\nonumber\\
&\hspace*{3.6cm}+\sum_{i=1}^{p}\big(F_\omega^*M_iP_{12,\omega}M_{k,i}+M_iP_{12,\omega}M_{k,i}F_{k,\omega}\big)=0,\label{eq:30}\\
&A_k^TQ_{k,\omega}+Q_{k,\omega}A_k+F_{k,\omega}^*C_k^TC_k+C_k^TC_kF_{k,\omega}\nonumber\\
&\hspace*{3.1cm}+\sum_{i=1}^{p}\big(F_{k,\omega}^*M_{k,i}P_{k,\omega}M_{k,i}+M_{k,i}P_{k,\omega}M_{k,i}F_{k,\omega}\big)=0.\label{eq:31}
\end{align}
Finally, the $\mathcal{H}_{2,\omega}$ norm of $G-G_k$ can be expressed as:
\begin{align}
||G-G_k||_{\mathcal{H}_{2,\omega}}&=\sqrt{\operatorname{trace}(B_e^TQ_{e,\omega}B_e)}\nonumber\\
&=\sqrt{\operatorname{trace}(B^TQ_\omega B-2B^TQ_{12,\omega}B_k+B_k^TQ_{k,\omega}B_k)}.\nonumber
\end{align}
\begin{corollary}
The expression $||G-G_k||_{\mathcal{H}_{2,\omega}}^2=||G||_{\mathcal{H}_{2,\omega}}^2-2\langle  G , G_k \rangle_{\mathcal{H}_{2,\omega}}+||G_k||_{\mathcal{H}_{2,\omega}}^2$ holds, where $\langle  G , G_k \rangle_{\mathcal{H}_{2,\omega}}$ denotes the $\mathcal{H}_{2,\omega}$ inner product of $G$ and $G_k$.
\end{corollary}
\begin{proof}
The first and last terms in the expression for $||G-G_k||_{\mathcal{H}_{2,\omega}}^2$ are straightforward. The main objective is to demonstrate that the middle term corresponds to the $\mathcal{H}_{2,\omega}$ inner product of $G$ and $G_k$. By expanding the definition of the inner product, we can express it as:
\begin{align}
\langle  G , G_k \rangle_{\mathcal{H}_{2,\omega}}&=\frac{1}{2\pi}\int_{-\omega}^{\omega}\operatorname{trace}\big(G_1^*(j\nu)G_{k,1}(j\nu)d\nu\big)\nonumber\\
&+\frac{1}{4\pi^2}\operatorname{trace}\Big(\int_{-\omega}^{\omega}\int_{-\omega}^{\omega}\sum_{i=1}^{p}\big(G_{2,i}^*(j\nu_1,j\nu_2)G_{k,2,i}(j\nu_1,j\nu_2)\big)d\nu_1d\nu_2\Big).\nonumber
\end{align}
Furthermore,
\begin{align}
&\operatorname{trace}\Big(\frac{1}{2\pi}\int_{-\omega}^{\omega}G_1^*(j\nu)G_{k,1}(j\nu)d\nu\Big)\nonumber\\
&=\operatorname{trace}\Big(B^T\big(\frac{1}{2\pi}\int_{-\omega}^{\omega}(j\nu I-A)^{-*}C^TC_k(j\nu I-A_k)^{-1}d\nu\big)B_k\Big),\nonumber\\
&\operatorname{trace}\Big(\frac{1}{4\pi^2}\int_{-\omega}^{\omega}\int_{-\omega}^{\omega}\sum_{i=1}^{p}G_{2,i}^*(j\nu_1,j\nu_2)G_{k,2,i}(j\nu_1,j\nu_2)d\nu_1d\nu_2\Big)\nonumber\\
&=\operatorname{trace}\Bigg(B^T\Big[\frac{1}{4\pi^2}\int_{-\omega}^{\omega}(j\nu_1 I-A)^{-*}\Big(\sum_{i=1}^{p}M_i\big(\int_{-\omega}^{\omega}(j\nu_2 I-A)^{-1}B\nonumber\\
&\hspace*{4cm}B_k^T(j\nu_2 I-A_k)^{-*}d\nu_2\big)M_{k,i}\Big)(j\nu_1 I-A_k)^{-1}\Big]B_k\Bigg).\nonumber
\end{align}
The Sylvester equations (\ref{eq:26}) and (\ref{eq:28}) can be solved by evaluating the following integrals:
\begin{align}
Y_{12,\omega}&=\frac{1}{2\pi}\int_{-\omega}^{\omega}(j\nu I-A)^{-*}C^TC_k(j\nu I-A_k)^{-1}d\nu,\nonumber\\
Z_{12,\omega}&=\frac{1}{4\pi^2}\int_{-\omega}^{\omega}(j\nu_1 I-A)^{-*}\Bigg(\sum_{i=1}^{p}M_i\Big(\int_{-\omega}^{\omega}(j\nu_2 I-A)^{-1}B\nonumber\\
&\hspace*{4cm}B_k^T(j\nu_2 I-A_k)^{-*}\Big)M_{k,i}\Bigg)(j\nu_1 I-A_k)^{-1}d\nu_1;\nonumber
\end{align}cf. \cite{benner2021gramians,song2024balanced} Consequently, the $\mathcal{H}_{2,\omega}$ inner product between $G$ and $G_k$ can be written as
\begin{align}
\langle  G , G_k \rangle_{\mathcal{H}_{2,\omega}}=\operatorname{trace}\big(B^T(Y_{12,\omega}+Z_{12,\omega})B_k\big)=\operatorname{trace}(B^TQ_{12,\omega}B_k).\nonumber
\end{align}
\end{proof}
\subsection{Optimality Conditions}
In this subsection, we present the necessary conditions for achieving a local optimum of $||G-G_k||_{\mathcal{H}_{2,\omega}}^2$. These optimality conditions require the introduction of several new variables. We begin by defining $\bar{Z}_{12}$ and $\bar{Z}_k$ as the solutions to the following equations:
\begin{align}
A^T\bar{Z}_{12}+\bar{Z}_{12}A_k+\sum_{i=1}^{p}M_iP_{12,\omega}M_{k,i}=0,\nonumber\\
A_k^T\bar{Z}_k+\bar{Z}_kA_k+\sum_{i=1}^{p}M_{k,i}P_{k,\omega}M_{k,i}=0.\nonumber
\end{align}
It is important to note that $P_{12,\omega}$, $P_{k,\omega}$, $Z_{12,\omega}$ and $Z_{k,\omega}$ can be derived from $P_{12}$, $P_k$, $\bar{Z}_{12}$ and $\bar{Z}_k$, respectively, by restricting the integration limits to $[0,\omega]$ rad/sec in their integral definitions; see \cite{song2024balanced} for details. Next, we define $\tilde{P}_{12}$, $\tilde{P}_k$, $\tilde{Z}_{12}$, and $\tilde{Z}_k$ as follows:
\begin{align}
\tilde{P}_{12}&=P_{12}\Big|_{-\infty}^{\infty}-P_{12}\Big|_{0}^{\omega}=P_{12}-P_{12,\omega},\label{nst46}\\
\tilde{P}_k&=P_k\Big|_{-\infty}^{\infty}-P_k\Big|_{0}^{\omega}=P_k-P_{k,\omega},\label{nst47}\\
\tilde{Z}_{12}&=\bar{Z}_{12}\Big|_{-\infty}^{\infty}-\bar{Z}_{12}\Big|_{0}^{\omega}=\bar{Z}_{12}-Z_{12,\omega},\label{nst48}\\
\tilde{Z}_k&=\bar{Z}_k\Big|_{-\infty}^{\infty}-\bar{Z}_k\Big|_{0}^{\omega}=\bar{Z}_k-Z_{k,\omega}.\label{nst49}
\end{align}
Additionally, we define $V$, $W$ and $L_\omega$ as follows:
\begin{align}
V&=B_kB^T\bar{Z}_{12}-B_kB_k^T\bar{Z}_k+P_{12}^TC^TC_k-P_kC_kC_k^T\nonumber\\
&\hspace*{2cm}+P_{12}^T\sum_{i=1}^{p}M_iP_{12,\omega}M_{k,i}-P_k\sum_{i=1}^{p}M_{k,i}P_{k,\omega}M_{k,i},\nonumber\\
W&=\frac{j}{2\pi}\mathcal{L}(-j\nu I-A_k,V),\nonumber\\
L_\omega&=-Q_{12,\omega}^*\tilde{P}_{12}-\tilde{Z}_{12}^*P_{12,\omega}+Q_{k,\omega}\tilde{P}_k+\tilde{Z}_kP_{k,\omega}+W^*,\nonumber
\end{align}where $\mathcal{L}(-j\nu I-A_k,V)$ denotes the Fr\'{e}chet derivative of the matrix logarithm $ln(-j\nu I-A_k)$ in the direction of the matrix $V$, specifically:
\begin{align}
\mathcal{L}(-j\nu I-A_k,V)&=\int_{0}^{1}\big(\nu(-j\nu I-A_k-I)+I\big)^{-1}V\nonumber\\
&\hspace*{4cm}\big(\nu(-j\nu I-A_k-I)+I\big)^{-1}d\nu;\nonumber
\end{align}cf. \cite{higham2008functions}.

We are now ready to state the necessary conditions for a local optimum of $||G-G_k||_{\mathcal{H}_{2,\omega}}^2$.
\begin{theorem}\label{th1}
The local optimum of $||G-G_k||_{\mathcal{H}_{2,\omega}}^2$ must satisfy the following necessary conditions:
\begin{align}
-(Y_{12,\omega}+2Z_{12,\omega})^*P_{12,\omega}+(Y_{k,\omega}+2Z_{k,\omega})P_{k,\omega}+L_\omega&=0,\label{op1}\\
-P_{12,\omega}^*M_iP_{12,\omega}+P_{k,\omega}M_{k,i}P_{k,\omega}&=0,\label{op2}\\
-(Y_{12,\omega}+2Z_{12,\omega})^*B+(Y_{k,\omega}+2Z_{k,\omega})B_k&=0,\label{op3}\\
-CP_{12,\omega}+C_kP_{k,\omega}&=0.\label{op4}
\end{align}
\end{theorem}
\begin{proof}
The proof of this theorem is lengthy and complex, so it is provided in the Appendix.
\end{proof}
\subsection{Comparison with Local Optimum of $||G-G_k||_{\mathcal{H}_2}^2$}
In this subsection, we compare the necessary conditions for the local optima of $||G-G_k||_{\mathcal{H}_2}^2$ and $||G-G_k||_{\mathcal{H}_{2,\omega}}^2$.
To begin, we provide the expression for $||G-G_k||_{\mathcal{H}_2}$ as presented in \cite{benner2021gramians}. The controllability Gramian $P_e$ and the observability Gramian $Q_e=Y_e+Z_e$ of realization $(A_e,B_e,C_e,M_{e,1},\cdots,M_{e,p})$ can be computed by solving the following Lyapunov equations:
\begin{align}
A_eP_e+P_e A_e^T+ B_eB_e^T&=0,\nonumber\\
A_e^TY_e+Y_e A_e+C_e^TC_e&=0,\nonumber\\
A_e^TZ_e+Z_e A_e+\sum_{i=1}^{p}M_{e,i}P_e M_{e,i}&=0,\nonumber\\
A_e^TQ_e+Q_e A_e+C_e^TC_e+\sum_{i=1}^{p}M_{e,i}P_eM_{e,i}&=0.\nonumber
\end{align}
We then partition $P_e$, $Y_e$, $Z_e$, and $Q_e$ according to (\ref{partreal}) as follows:
\begin{align}
P_e&=\begin{bmatrix}P&P_{12}\\P_{12}^T&P_k\end{bmatrix},&
Y_e&=\begin{bmatrix}Y&-Y_{12}\\-Y_{12}^T&Y_k\end{bmatrix},\nonumber\\
Z_e&=\begin{bmatrix}Z&-Z_{12}\\-Z_{12}^T&Z_k\end{bmatrix},&
Q_e&=\begin{bmatrix}Q&-Q_{12}\\-Q_{12}^T&Q_k\end{bmatrix}.\nonumber
\end{align}
The $\mathcal{H}_2$ norm of $G-G_k$ can be expressed as:
\begin{align}
||G-G_k||_{\mathcal{H}_2}&=\sqrt{\operatorname{trace}(B_e^TQ_eB_e)}\nonumber\\
&=\sqrt{\operatorname{trace}(B^TQ B-2B^TQ_{12}B_k+B_k^TQ_kB_k)}.\nonumber
\end{align}
The optimality conditions (\ref{op1})-(\ref{op4}) and (\ref{op01})-(\ref{op04}) are similar, but there are some important differences. By restricting the integration limit of $P_e$ and $Q_e$ to $[0,\omega]$ rad/sec, the optimality conditions (\ref{op2})-(\ref{op4}) can be derived from (\ref{op02})-(\ref{op04}), respectively. However, the optimality condition (\ref{op01}) does not simplify to (\ref{op1}) by merely limiting the integration range.

Furthermore, from the optimality conditions (\ref{op02})-(\ref{op04}), we can deduce the optimal selections for $M_{k,i}$, $B_k$, and $C_k$ as:
\begin{align}
M_{k,i}&=P_k^{-1}P_{12}^TM_iP_{12}P_k^{-1},\label{ocm0}\\
B_k&=(Y_k+2Z_k)^{-1}(Y_{12}+2Z_{12})^TB,\label{obm0}\\
C_k&=CP_{12}P_k^{-1}.\label{occ0}
\end{align}
By restricting the integration limits of $P_e$ and $Q_e$ to $[0,\omega]$ rad/sec, we can derive the frequency-limited optimal choices for $M_{k,i}$, $B_k$, and $C_k$ from (\ref{ocm0})-(\ref{occ0}) as follows:
\begin{align}
M_{k,i}&=P_{k,\omega}^{-1}P_{12,\omega}^*M_iP_{12,\omega}P_{k,\omega}^{-1},\label{ocm1}\\
B_k&=(Y_{k,\omega}+2Z_{k,\omega})^{-1}(Y_{12,\omega}+2Z_{12,\omega})^*B,\label{obm1}\\
C_k&=CP_{12,\omega}P_{k,\omega}^{-1}.\label{occ1}
\end{align}
The optimal projection matrices $V_k$ and $W_k$ for computing a local optimum of $||G-G_k||_{\mathcal{H}_2}^2$ are given by:
\begin{align}
V_k&=P_{12}P_k^{-1},&W_k&=(Y_{12}+2Z_{12})(Y_k+2Z_k)^{-1}.\nonumber
\end{align}
In the frequency-limited scenario, by setting:
\begin{align}
V_k&=P_{12,\omega}P_{k,\omega}^{-1},&W_k&=(Y_{12,\omega}+2Z_{12,\omega})(Y_{k,\omega}+2Z_{k,\omega})^{-1},\nonumber
\end{align} we make the optimal choices for $M_{k,i}$, $B_k$, and $C_k$ as indicated by the optimality conditions (\ref{op2})-(\ref{op4}). However, with this choice of $V_k$ and $W_k$, determining an optimal $A_k$ remains elusive. By enforcing the Petrov–Galerkin projection condition $W_k^*V_k=I$, we ensure
\begin{align}
-(Y_{12,\omega}+2Z_{12,\omega})^*P_{12,\omega}+(Y_{k,\omega}+2Z_{k,\omega})P_{k,\omega}=0.\nonumber
\end{align}It is important to note that, generally, $L_\omega$ does not simplify to zero with this choice of projection matrices. Consequently, this selection introduces a deviation in the optimality condition (\ref{op1}) quantified by $L_\omega$. In contrast, in the classical $\mathcal{H}_2$-optimal MOR framework, enforcing the Petrov–Galerkin projection condition $W_k^TV_k=I$ satisfies the optimality condition (\ref{op01}) and achieves a local optimum of $||G-G_k||_{\mathcal{H}_{2,\omega}}^2$. In summary, it is generally impossible to attain a local optimum of $||G-G_k||_{\mathcal{H}_{2,\omega}}^2$ within the projection framework. While the optimality conditions (\ref{op2})-(\ref{op4}) can be precisely met, the optimality condition (\ref{op1}) can only be approximately satisfied.

Up to this point, we have determined the appropriate projection matrices $V_k=P_{12,\omega}P_{k,\omega}^{-1}$ and $W_k=(Y_{12,\omega}+2Z_{12,\omega})(Y_{k,\omega}+2Z_{k,\omega})^{-1}$ for the problem at hand. However, these matrices depend on the ROM $(A_k,B_k,C_k,M_{k,1},\cdots,M_{k,p})$, which is unknown. Therefore, equations (\ref{eq:2}) and (\ref{eq:24})-(\ref{eq:29}) form a coupled system of equations, expressed as:
\begin{align}
(A_k,B_k,C_k,M_{k,1},\cdots,M_{k,p})&=f(P_{12,\omega},P_{k,\omega},Y_{12,\omega},Y_{k,\omega},Z_{12,\omega},Z_{k,\omega}),\nonumber\\
(P_{12,\omega},P_{k,\omega},Y_{12,\omega},Y_{k,\omega},Z_{12,\omega},Z_{k,\omega})&=g(A_k,B_k,C_k,M_{k,1},\cdots,M_{k,p}).\nonumber
\end{align}
The stationary points of the function
\begin{align}
(A_k,B_k,C_k,M_{k,1},\cdots,M_{k,p})=f\big(g(A_k,B_k,C_k,M_{k,1},\cdots,M_{k,p})\big)\nonumber
\end{align} satisfy the optimality conditions (\ref{op2})-(\ref{op4}). Additionally, by enforcing the Petrov-Galerkin projection condition $W_k^*V_k=I$, the optimality condition (\ref{op1}) is approximately satisfied, with the deviation quantified by $L_\omega$. In the classical $\mathcal{H}_2$-optimal MOR scenario, the situation is similar; however, enforcing the Petrov–Galerkin projection condition $W_k^TV_k=I$ at the stationary points ensures that all the optimality conditions (\ref{op01})-(\ref{op04}) are fully satisfied.

In the classical $\mathcal{H}_2$-optimal MOR case, it is demonstrated that if the reduction matrices are chosen as $V_k=P_{12}$ and $W_k=Y_{12}+2Z_{12}$ instead of $V_k=P_{12}P_k^{-1}$ and $W_k=(Y_{12}+2Z_{12})(Y_k+2Z_k)^{-1}$, the stationary points satisfy:
\begin{align}
P_k&=I,\hspace*{0.5cm}\textnormal{and}\hspace*{0.5cm}Y_k+2Z_k=I.\nonumber
\end{align}Thus, the projection matrices $V_k=P_{12}$ and $W_k=Y_{12}+2Z_{12}$, along with the Petrov–Galerkin projection condition $W_k^TV_k=I$, satisfy all the optimality conditions (\ref{op01})-(\ref{op04}). However, using $V_k=P_{12,\omega}$ and $W_k=Y_{12,\omega}+2Z_{12,\omega}$ along with the Petrov–Galerkin projection condition $W_k^*V_k=I$ does not satisfy any of the optimality conditions (\ref{op1})-(\ref{op4}).
\begin{theorem}
If $W_k^*F_\omega B=F_{k,\omega}B_k$, $CF_\omega V_k=C_kF_{k,\omega}$, and $V_k^*M_iF_\omega V_k=M_{k,i}F_{k,\omega}$, then selecting $V_k=P_{12,\omega}$ and $W_k=Y_{12,\omega}+2Z_{12,\omega}$, together with the Petrov–Galerkin projection condition $W_k^*V_k=I$ ensures that:
\begin{align}
P_{k,\omega}&=I,\hspace*{0.5cm}\textnormal{and}\hspace*{0.5cm}Y_{k,\omega}+2Z_{k,\omega}=I,\label{eq:64}
\end{align} which in turn satisfies the optimality conditions (\ref{op2})-(\ref{op4}).
\end{theorem}
\begin{proof}
By pre-multiplying (\ref{eq:24}) with $W_k^*$, we obtain:
\begin{align}
W_k^*\big(AP_{12,\omega}+P_{12,\omega}A_k^T+F_\omega BB_k^T+BB_k^TF_{k,\omega}^*\big)&=0\nonumber\\
A_k+A_k^T+F_{k,\omega}B_kB_k^T+B_kB_k^TF_{k,\omega}&=0.\nonumber
\end{align} Given the uniqueness of solution for equation (\ref{eq:25}), we have $P_{k,\omega}=I$.

Note that $Y_{12,\omega}+2Z_{12,\omega}$ and $Y_{k,\omega}+2Z_{k,\omega}$ satisfies the following equations:
\begin{align}
&A^T(Y_{12,\omega}+2Z_{12,\omega})+(Y_{12,\omega}+2Z_{12,\omega})A_k+F_\omega CC_k^T+CC_k^TF_{k,\omega}^*\nonumber\\
&\hspace*{2.95cm}+2\sum_{i=1}^{p}\big(F_\omega^*M_iP_{12,\omega}M_{k,i}+M_iP_{12,\omega}M_{k,i}F_{k,\omega}\big)=0,\label{lmeq:35}\\
&A_k^T(Y_{k,\omega}+2Z_{k,\omega})+(Y_{k,\omega}+2Z_{k,\omega})A_k+F_{k,\omega} C_kC_k^T+C_kC_k^TF_{k,\omega}^*\nonumber\\
&\hspace*{2.95cm}+2\sum_{i=1}^{p}\big(F_{k,\omega}^*M_{k,i}P_{k,\omega}M_{k,i}+M_{k,i}P_{k,\omega}M_{k,i}F_{k,\omega}\big)=0.\label{lmeq:36}
\end{align}
Taking the Hermitian of equation (\ref{lmeq:35}) and post-multiplying it by $V_k$, we get:
\begin{align}
&\Big(A_k^T(Y_{12,\omega}+2Z_{12,\omega})^*+(Y_{12,\omega}+2Z_{12,\omega})^*A+C_k^TCF_\omega+F_{k,\omega}^*C_k^TC\nonumber\\
&\hspace*{3.3cm}+\sum_{i=1}^{p}2M_{k,i}P_{12,\omega}^*M_iF_\omega+\sum_{i=1}^{p}2F_{k,\omega}^*M_{k,i}P_{12,\omega}^*M_i\Big)V_k\nonumber\\
&=A_k^T+A_k+C_k^TC_kF_{k,\omega}+F_{k,\omega}^*C_k^TC_k\nonumber\\
&\hspace*{3.3cm}+\sum_{i=1}^{p}2M_{k,i}M_{k,i}F_{k,\omega}+\sum_{i=1}^{p}2F_{k,\omega}^*M_{k,i}M_{k,i}=0.\nonumber
\end{align}With the uniqueness of solutions for equations (\ref{eq:25}) and (\ref{lmeq:36}), we have $P_{k,\omega}=I$ and $Y_{k,\omega}+2Z_{k,\omega}=I$. As a result, the optimality conditions (\ref{op2})-(\ref{op4}) are met with $V_k=P_{12,\omega}$, and $W_k=Y_{12,\omega}+2Z_{12,\omega}$.
\end{proof}
In general, choosing $V_k=P_{12,\omega}$ and $W_k=Y_{12,\omega}+2Z_{12,\omega}$ while enforcing the Petrov–Galerkin projection condition $W_k^*V_k=I$ does not meet the conditions $W_k^*F_\omega B=F_{k,\omega}B_k$, $CF_\omega V_k=C_kF_{k,\omega}$, and $V_k^TM_iF_\omega V_k=M_{k,i}F_{k,\omega}$. Consequently, the condition (\ref{eq:64}) is not satisfied, and therefore, optimality conditions (\ref{op2})-(\ref{op4}) are not fulfilled.

The rationale for using $V_k = P_{12,\omega}$ and $W_k = Y_{12,\omega} + 2Z_{12,\omega}$ instead of $V_k = P_{12,\omega} P_{k,\omega}^{-1}$ and $W_k = (Y_{12,\omega} + 2Z_{12,\omega})(Y_{k,\omega} + 2Z_{k,\omega})^{-1}$ from a Petrov-Galerkin projection perspective is that both sets span the same subspace and yield the same ROM but with different state-space realizations \cite{gallivan2004sylvester}. In the next section, it will be demonstrated that eliminating $P_{k,\omega}^{-1} \) and \( (Y_{k,\omega} + 2Z_{k,\omega})^{-1}$ does not result in any noticeable loss of accuracy, and the deviations from the optimality conditions are similarly insignificant. Therefore, numerically, $V_k = P_{12,\omega}$ and $W_k = Y_{12,\omega} + 2Z_{12,\omega}$ are valid choices, even though they are theoretically inferior to $V_k = P_{12,\omega}P_{k,\omega}^{-1}$ and $W_k = (Y_{12,\omega} + 2Z_{12,\omega})(Y_{k,\omega} + 2Z_{k,\omega})^{-1}$.
\subsection{Avoiding Complex Matrices}
So far, for simplicity, we have considered $V_k$ and $W_k$ as complex matrices in the problem under consideration. However, in practice, using complex projection matrices results in a complex ROM, which is undesirable since most practical dynamical systems are represented by real-valued mathematical models. This issue can be addressed by extending the desired frequency interval to include negative frequencies, i.e., $[-\omega,0]$ rad/sec. Additionally, we have assumed that the desired frequency interval starts from $0$ rad/sec for simplicity. For any general frequency interval $[-\omega_2,-\omega_1]\cup[\omega_1,\omega_2]$ rad/sec, the only modification needed is in the computation of $F_\omega$ and $F_{k,\omega}$ given by:
\begin{align}
F_\omega&=Re\Big(\frac{j}{\pi}ln\big((j\omega_1 I+A\big)^{-1}(j\omega_2 I+A)\big)\Big),\label{ex:66}\\
F_{k,\omega}&=Re\Big(\frac{j}{\pi}ln\big((j\omega_1 I+A_k\big)^{-1}(j\omega_2 I+A_k)\big)\Big);\label{ex:67}
\end{align}see \cite{petersson2013nonlinear} for more details.
\subsection{Approximating Matrix Logarithm $F_\omega$}
The matrix logarithm \( F_\omega \) is commonly encountered in most frequency-limited MOR algorithms. However, computing \( F_\omega \) can become computationally expensive when the order \( n \) of the original model is large. In such cases, Krylov subspace-based methods, as proposed in \cite{benner2016frequency}, can be used to approximate \( F_\omega B \) rather than computing \( F_\omega \) directly. The key idea of this method is to project \( (A, B) \) onto a reduced subspace to obtain \( (A_k, B_k) \), and then approximate \( F_\omega B \) as \( F_\omega B \approx VF_{k,\omega}B_k \). The product \( F_\omega^* C^T \) can be similarly approximated using this approach. However, this approach cannot be used to approximate \( F_\omega^* M_i \), since \( M_i \) is not a thin matrix like \( B \) and \( C \), where \( m \ll n \) and \( p \ll n \), respectively. Instead, if \( F_\omega^* M_i P_{12,\omega} \) is approximated, the method in \cite{benner2016frequency} can still be applied. However, this approximation of \( F_\omega^* M_i P_{12,\omega} \) needs to be updated in every iteration. Since the method in \cite{benner2016frequency} is also iterative, we propose a different, non-iterative approach, as explained in the sequel.

 To better understand the matrix logarithm and its computation, let us look at its origins. In frequency-weighted MOR, weights are used to emphasize the desired frequency ranges where superior accuracy is required. When these weights are ideal bandpass filters with the desired frequency range as the passband, they appear as matrix logarithms in the integral expressions for the Gramians. For more details on different ideal bandpass filter choices in frequency-weighted MOR, see \cite{wortelbore1994frequency}. A detailed discussion on the relationship between frequency-weighted BT and FLBT can be found in \cite{gugercin2004survey}. It is observed in \cite{gugercin2004survey}, \cite{zulfiqar2022frequency}, and \cite{zulfiqar2021weighted} that as the order of the bandpass filter increases, the numerical differences between frequency-weighted and frequency-limited MOR diminish, even for low-order filters. This is because the frequency-weighted Gramians become almost the same as the frequency-limited Gramians. Therefore, an ideal (and impractical) filter is not really needed. However, this idea was not applied in \cite{benner2016frequency} for approximating $F_\omega B$ and $F_\omega^* C^T$. Designing low-order analog filters is simpler and more straightforward using the analytical expressions of standard filter designs like Butterworth or Chebyshev filters, or MATLAB functions such as \textit{`butter'}, \textit{`cheby1'}, and \textit{`cheby2'}. By using a practical, low-order filter instead of an ideal one, we can avoid the need to compute the matrix logarithm, making the process simpler.

Let \( v(s) = c_v(sI - a_v)^{-1} b_v \) be an \( n_v^{th} \)-order bandpass filter with a passband in the range \( (\omega_1, \omega_2) \) rad/sec. The controllability Gramian \( p_v \) of the pair \( (a_v, b_v) \) satisfies the following Lyapunov equation:
\begin{align}
a_v p_v + p_v a_v^T + b_v b_v^T = 0. \label{smalllyap}
\end{align}
Next, define a frequency weight \( V(s) \) as
\[
V(s) = v(s)I_m = C_v(sI - A_v)^{-1} B_v,
\]
where
\[
A_v = I_m \otimes a_v, \quad B_v = I_m \otimes b_v, \quad C_v = I_m \otimes c_v.
\]
Additionally, the controllability Gramian \( P_v \) of the pair \( (A_v, B_v) \) is related to \( p_v \) by $P_v = I_m \otimes p_v$.

Next, we define the matrices $\mathscr{A}$, $\hat{\mathscr{A}}$, $\mathscr{B}$, and $\hat{\mathscr{B}}$ as follows:
 \begin{align}
 \mathscr{A}&=\begin{bmatrix}A&BC_v\\0&A_v\end{bmatrix},&\hat{\mathscr{A}}&=\begin{bmatrix}A_k&B_kC_v\\0&A_v\end{bmatrix},\nonumber\\
 \mathscr{B}&=\begin{bmatrix}0\\B_v\end{bmatrix},&\hat{\mathscr{B}}&=\begin{bmatrix}0\\B_v\end{bmatrix}.\nonumber
 \end{align}Due to triangular structure of $\mathscr{A}$ and $\hat{\mathscr{A}}$, the following expression holds:
 \begin{align}
 &\frac{1}{2\pi}\int_{-\infty}^{\infty}(j\nu I-\mathscr{A})^{-1}\mathscr{B}\hat{\mathscr{B}}^T(j\nu I-\hat{\mathscr{A}})^{-*}d\nu\label{int_exp}\\
 =&\begin{bmatrix}\frac{1}{2\pi}\int_{-\infty}^{\infty}(j\nu I-A)^{-1}BV(j\nu)V^*(\nu)B_k^T(j\nu I-A_k)^{-*}d\nu&\star\\\star&\star\end{bmatrix}\nonumber\\
 =&\begin{bmatrix}\mathscr{P}_{12,\omega}&\star\\\star&\star\end{bmatrix}.\nonumber
 \end{align}It can be observed that when \( v(s) \) (and thus \( V(s) \)) is an ideal bandpass filter, i.e., \( v(j\omega) = 1 \) within the range \( \omega_1 \) to \( \omega_2 \) rad/sec, and \( v(j\omega) = 0 \) outside this range, we have \( \mathscr{P}_{12,\omega} = P_{12,\omega} \). When \( v(s) \) is not an ideal filter, we have \( \mathscr{P}_{12,\omega} \approx P_{12,\omega} \). The accuracy of the approximation \( \mathscr{P}_{12,\omega} \approx P_{12,\omega} \) depends on the quality of the bandpass filter \( v(s) \), which can be improved by increasing the order of the filter. In the next section, it will be shown that the numerical difference between \( \mathscr{P}_{12,\omega} \) and \( P_{12,\omega} \) nearly disappears for a Butterworth filter of order $16$ to $20$. Additionally, due to the connection between the Sylvester equation and the integral expression, equation (\ref{int_exp}) solves the following Sylvester equation: 
 \begin{align}
 \mathscr{A}\begin{bmatrix}\mathscr{P}_{12,\omega}&\star\\\star&\star\end{bmatrix}+\begin{bmatrix}\mathscr{P}_{12,\omega}&\star\\\star&\star\end{bmatrix}\hat{\mathscr{A}}^T+\mathscr{B}\hat{\mathscr{B}}^T&=0.\label{sylv_1}
 \end{align}Thus, \( P_{12,\omega} \) can be closely approximated by solving the Sylvester equation (\ref{sylv_1}).

Let us define \( \hat{P}_v \) and \( \hat{P}_{k,v} \), which solve the following Sylvester equations:
 \begin{align}
 A\hat{P}_v+\hat{P}_vA_v^T+BC_vP_v&=0,\label{Bw}\\
 A_k\hat{P}_{k,v}+\hat{P}_{k,v}A_v^T+B_kC_vP_v&=0.
 \end{align} By expanding equation (\ref{sylv_1}), it can be easily observed that \( \mathscr{P}_{12,\omega} \) satisfies the following Sylvester equation:
 \begin{align}
 A\mathscr{P}_{12,\omega}+\mathscr{P}_{12,\omega}A_k^T+\hat{P}_vC_v^TB_k^T+BC_v\hat{P}_{k,v}^T&=0.
 \end{align} This is the same Sylvester equation that arises in frequency-weighted \( \mathcal{H}_2 \)-optimal MOR algorithms \cite{zulfiqar2022frequency, zulfiqar2021weighted}. As the numerical difference between \( \mathscr{P}_{12,\omega} \) and \( P_{12,\omega} \) diminishes, \( F_\omega B \) and \( B_k^T F_\omega^T \) become numerically equivalent to \( \hat{P}_v C_v^T \) and \( C_v \hat{P}_{k,v}^T \), respectively. The pseudo-code for the method presented here to approximate \( F_\omega B \) is summarized in Algorithm \ref{alg0}.
 \begin{algorithm}
\caption{Algorithm for Approximating $F_\omega B$\\$B_\omega=B_\Omega(A,B,\omega_1,\omega_2,n_v)$}
\begin{algorithmic}[1]\label{alg0}
\STATE Design an $n_v^{th}$-order bandpass filter $v(s)=c_v(sI-a_v)^{-1}b_v$ with passband $(\omega_1,\omega_2)$ rad/sec.
\STATE Solve the Lyapunov equation (\ref{smalllyap}) to compute $p_v$.
\STATE Set  $A_v=I_m\otimes a_v$, $C_v=I_m\otimes c_v$, and $P_v=I_m\otimes p_v$.\nonumber
\STATE Solve the Sylvester equation (\ref{Bw}) to compute $\hat{P}_v$.
\STATE Set $B_\omega=\hat{P}_vC_v^T\approx F_\omega B$.
\end{algorithmic}
\end{algorithm}
The product \( F_\omega^T C^T \) can also be approximated using Algorithm \ref{alg0} by substituting \( (A, B) \) with \( (A^T, C^T) \), i.e., $B_\Omega(A^T,C^T,\omega_1,\omega_2,n_v)$.

Let \( \mathscr{C}^T \) be defined as \( \mathscr{C}^T = \sum_{i=1}^{p} M_i P_{12,\omega} M_{k,i} \). The Sylvester equation (\ref{eq:28}) can then be rewritten as:

\[
A^T Z_{12,\omega} + Z_{12,\omega} A_k + F_\omega^T \mathscr{C}^T + \mathscr{C}^T F_{k,\omega} = 0.
\]
It can be observed that \( F_\omega^T \mathscr{C}^T \) can be approximated using Algorithm \ref{alg0} by replacing \( (A, B) \) with \( (A^T, \mathscr{C}^T) \). Since \( F_\omega^T \mathscr{C}^T \) depends on the ROM, it must be recomputed in each iteration, unlike \( F_\omega B \) and \( F_\omega^T C^T \), which only need to be computed once. Although \( F_{k,\omega} \) can be computed directly and inexpensively from the matrix logarithm (\ref{ex:67}), when approximations of \( F_\omega B \), \( F_\omega^T C^T \), and \( F_\omega^T \mathscr{C}^T \) are used, the theoretical connection with frequency-weighted MOR requires that \( F_{k,\omega} B_k \), \( F_{k,\omega}^T C_k^T \), and \( F_{k,\omega} I_k \) be replaced by their approximations obtained via Algorithm \ref{alg0}. In our experiments, using \( F_{k,\omega} \) exactly while replacing \( F_\omega B \), \( F_\omega^T C^T \), and \( F_\omega^T \mathscr{C}^T \) with their approximations from Algorithm \ref{alg0} resulted in inferior accuracy, which supports the theoretical findings. Therefore, it is recommended to also approximate \( F_{k,\omega} B_k \), \( F_{k,\omega}^T C_k^T \), and \( F_{k,\omega} I_k \) when approximations of \( F_\omega B \), \( F_\omega^T C^T \), and \( F_\omega^T \mathscr{C}^T \) from Algorithm \ref{alg0} are used.
\subsection{Algorithm}
We are now ready to introduce our proposed algorithm, referred to in this paper as the ``Frequency-limited \( \mathcal{H}_2 \) Near-optimal Iterative Algorithm (FLHNOIA).'' The pseudo-code for FLHNOIA is provided in Algorithm \ref{alg1}. The algorithm begins with an arbitrary initial guess for the ROM and iteratively refines it until convergence is reached. In Step \ref{step1}, \( F_\omega B \) and \( F_\omega^T C^T \) are approximated using Algorithm \ref{alg0}. The iterative refinement of the initial guess starts in Step \ref{step2}. In Step \ref{step5}, \( \mathscr{P}_{12,\omega} \) and \( \mathscr{X}_{12,\omega} \) are orthogonalized to prevent linear dependence and to enforce the Petrov-Galerkin condition \( W_k^T V_k = I \). It is important to note that the columns of \( V_k \) and \( W_k \) still span the same subspace as \( \mathscr{P}_{12,\omega} \) and \( \mathscr{X}_{12,\omega} \), respectively, and they only provide a different state-space realization of the same ROM \cite{gallivan2004sylvester}.
\begin{algorithm}
\caption{FLHNOIA}
\textbf{Input:} Full order system: $(A,B,C,M_1,\cdots,M_p)$; Desired frequency interval: $[\omega_1,\omega_2]$ rad/sec; Initial guess of ROM: $(A_k,B_k,C_k,M_{k,1},\cdots,M_{k,p})$.\\
 \textbf{Output:} ROM: $(A_k,B_k, C_k,M_{k,1}, \cdots,M_{k,p})$.
\begin{algorithmic}[1]\label{alg1}
\STATE Compute $B_\omega$ and $C_\omega^T$ using Algorithm \ref{alg0} as follows: $B_\omega=B_\Omega(A,B,\omega_1,\omega_2,n_v)$ and $C_\omega^T=B_\Omega(A^T,C^T,\omega_1,\omega_2,n_v)$. \label{step1}
\STATE \textbf{while}\big(not converged\big) \textbf{do} \label{step2}
\STATE Compute $F_{k,\omega}$ using Algorithm \ref{alg0} as follows: $B_\omega=B_\Omega(A_k,I,\omega_1,\omega_2,n_v)$.
\STATE Solve the following Sylvester equation to compute $\mathscr{P}_{12,\omega}$:

\(
A\mathscr{P}_{12,\omega}+\mathscr{P}_{12,\omega}A_k^T+B_\omega B_k^T+BB_kF_{k,\omega}^T=0.
\)
\STATE Set $\mathscr{C}^T=\sum_{i=1}^{p}M_i\mathscr{P}_{12,\omega}M_{k,i}$ and compute $\mathscr{C}_\omega^T$ using Algorithm \ref{alg0} as follows: $\mathscr{C}_\omega^T=B_\Omega(A^T,\mathscr{C}^T,\omega_1,\omega_2,n_v)$.
\STATE Solve the following Sylvester equation to compute $\mathscr{X}_{12,\omega}$:

\(
A^T\mathscr{X}_{12,\omega}+\mathscr{X}_{12,\omega}A_k+C_\omega^TC_k+C^TC_kF_{k,\omega}+2\mathscr{C}_\omega^T+2\mathscr{C}^TF_{k,\omega}=0.
\)
\STATE Set $V_k=orth(\mathscr{P}_{12,\omega})$, $R_k=orth(\mathscr{X}_{12,\omega})$, and $W_k=R_k(V_k^TR_k)^{-1}$.\label{step5}
\STATE Update the ROM as $A_k=W_k^TAV_k$, $B_k=W_k^TB$, $C_k=CV_k$, $M_{k,i}=V_k^TM_iV_k$.
\STATE \textbf{end while}
\end{algorithmic}
\end{algorithm}
 \subsection{Computational Cost Comparison with HOMORA}
A special type of Sylvester equation, known as the \textit{``sparse-dense''} Sylvester equation, appears in most \( \mathcal{H}_2 \)-optimal MOR algorithms, including HOMORA and FLHNOIA. This equation has the following structure:
\begin{align}
A\mathcal{C}+\mathcal{C}\mathcal{B}+\mathcal{D}&=0,\label{sparse_desnse_Sylv}
\end{align}where the large matrix \(A \in \mathbb{R}^{n \times n} \) is sparse, while the smaller matrix \( \mathcal{B} \in \mathbb{R}^{k \times k} \) is dense, with \( k \ll n \). The sparsity in \( A \) arises from the mathematical modeling of high-order systems. An efficient algorithm to solve sparse-dense Sylvester equations is described in \cite{MPIMD11-11}. The main reasons for its computational efficiency are summarized as follows. If the matrix \( \mathcal{B} \) is replaced by its eigenvalue decomposition \( \mathcal{B} = R S R^{-1} \), the matrix \( \mathcal{C} \) can be solved column-wise, requiring \( k \) linear solves of the form \( A x = b \). However, as noted in \cite{MPIMD11-11}, the multiplication by the inverse of \( R \) in the last step can lead to numerical issues. To mitigate this, it is recommended to replace \( \mathcal{B} \) with its Schur decomposition \( \mathcal{B} = R S R^* \), which avoids the need for inversion. The triangular structure of \( S \) still allows for column-wise computation of \( \mathcal{C} \), requiring \( k \) linear solves of \( A x = b \). By using efficient sparse linear solvers, such as those described in \cite{demmel1999asynchronous, demmel1999supernodal, davis2004algorithm}, the sparse-dense Sylvester equation can be solved efficiently. In summary, the small size of \( \mathcal{B} \) limits the number of linear solves of \( A x = b \), while the sparsity of \( A \) allows the use of computationally efficient sparse linear solvers, leading to an overall efficient solution of the sparse-dense Sylvester equation.

The primary computational burden in HOMORA lies in solving two sparse-dense Sylvester equations during each iteration. In total, this results in \( 2ki \) linear solves of the form \( Ax = b \) across \( i \) iterations. As long as the number of iterations \( i \) remains small, the computational time in HOMORA is manageable. Similarly, in FLHNOIA, the main computational load comes from solving sparse-dense Sylvester equations. In Algorithm \ref{alg0}, the filter design can be performed on-the-fly using either analytical expressions for analog filters or MATLAB’s built-in functions. Since the order \( n_v \) is small, \( p_v \) can also be computed cheaply. The primary computational burden in Algorithm \ref{alg0} is solving the sparse-dense Sylvester equation (\ref{Bw}). Algorithm \ref{alg0} is used twice in Step 1 of FLHNOIA, resulting in a total of \( (m + p)n_v \) linear solves \( Ax = b \). In each iteration of FLHNOIA, three sparse-dense Sylvester equations are solved, totaling \( (2k + n_vk)i \) linear solves over \( i \) iterations. Consequently, compared to HOMORA, FLHNOIA requires \( (m + p + ki)n_v \) additional linear solves of the form \( Ax = b \). Since \( m \ll n \), \( p \ll n \), \( k \ll n \), and \( n_v \ll n \), this added computational burden remains acceptable as long as the number of iterations \( i \) is small. Nevertheless, FLHNOIA is less efficient compared to HOMORA, and \( (m + p + ki)n_v \) more linear solves is the computational price that needs to be paid for addressing frequency-limited \( \mathcal{H}_2 \)-optimal MOR for LQO systems.
\begin{remark}
Algorithm \ref{alg0} can be used to approximate \( F_\omega \) by passing \( (A, I) \), in a similar manner to how \( F_{k,\omega} \) is computed in Algorithm \ref{alg1} using Algorithm \ref{alg0}. However, in this case, the Sylvester equation (\ref{Bw}) will be large-scale rather than the tall-skinny form seen in the sparse-dense Sylvester equation (\ref{sparse_desnse_Sylv}). Specifically, the matrix \( A_v \) will have dimensions \( n_vn \times n_vn \), which is larger than \( A \). In our experiments, we examined the decay of singular values of \( \hat{P}_v \) and found it to be numerically low-rank, although the decay was not as rapid as observed for frequency-limited Gramians in \cite{benner2016frequency, song2024balanced}. We then generated low-rank approximations of \( \hat{P}_v \) using truncated singular value decomposition and studied their impact on the quality of \( F_\omega \) approximation. We noted that beyond a certain rank, there was no further improvement in the quality of \( F_\omega \) approximation. This observation suggests that a computationally efficient low-rank approximation of this Sylvester equation can be obtained using the ADI method presented in \cite{benner2014computing}. The shifts in this method can be automatically determined using the strategy proposed in \cite{benner2014self}. If the ADI method converges and provides a good approximation of \( \hat{P}_v \), then \( F_\omega \) can be computed efficiently. The potential advantage of this approach is that, once \( F_\omega \) is approximated efficiently and accurately, FLHNOIA will require the same number of linear solves \( 2ki \) in \( i \) iterations as HOMORA. While this paper focuses on sparse-dense Sylvester equations and does not explore ADI methods, the approach outlined above offers a viable alternative, and users may consider it for approximating \( F_\omega \).
\end{remark}
\subsection{Choice of Initial Guess}
In general, the iterative refinement process in fixed-point iteration-based \( \mathcal{H}_2 \)-optimal MOR algorithms is quite robust. Even when initialized with an arbitrary guess, it quickly refines the ROM and achieves good accuracy. However, as the number of inputs and outputs increases, these algorithms begin to face difficulties, and a good initial guess becomes more important. The projection matrix \( V_k \) in FLIRKA \cite{zulfiqar2023frequency} for standard LTI systems and in the proposed FLHNOIA is common. Therefore, it is reasonable to follow the same guidelines for the initial guess as used in FLIRKA \cite{zulfiqar2023frequency}. These guidelines recommend selecting an initial guess that includes the dominant modes and zeros of \( A \). Such an initial guess can be efficiently generated using the algorithms presented in \cite{rommes2006efficient,martins2007computation}.
\subsection{Stopping Criterion and Error Monitoring}
FLHNOIA, like HOMORA, is a fixed-point iteration algorithm. Consequently, as with HOMORA, convergence is not guaranteed in FLHNOIA. As discussed in the previous subsection, the computational efficiency of the algorithm depends on the number of iterations \( i \) required to achieve the final ROM. Therefore, a stopping criterion is necessary for FLHNOIA. Since both standard LTI systems and LQO systems share the transfer function \( G_{k,1}(s) \), one possible stopping criterion could be the stagnation of the eigenvalues of \( A_k \), based on its connection to the standard \( \mathcal{H}_{2,\omega} \)-optimality conditions derived in \cite{vuillemin2014thesis}. This criterion is widely used in \( \mathcal{H}_2 \)-optimal MOR algorithms, including \cite{vuillemin2013h2,zulfiqar2023frequency}. It is important to note that the optimality condition in (\ref{op3}) is common to both standard LTI systems \cite{petersson2014model,vuillemin2013h2} and LQO systems, providing a solid theoretical foundation for this stopping criterion. While computing the \( \mathcal{H}_{2,\omega} \) norm of the error \( G-G_k \) is computationally expensive and impractical for tracking the error, an alternative approach can be used. Specifically, if \( Q_{k,\omega} \) is computed in every iteration (which is computationally inexpensive), and instead of directly computing \( \mathscr{X}_{12,\omega} \), the terms \( Y_{12,\omega} \) and \( Z_{12,\omega} \) are computed separately, then in each iteration we can evaluate both \( Q_{12,\omega} = Y_{12,\omega} + Z_{12,\omega} \) and \( \mathscr{X}_{12,\omega} = Y_{12,\omega} + 2Z_{12,\omega} \). It is important to note that the terms of \( ||G-G_k||_{\mathcal{H}_{2,\omega}} \) that change in every iteration  are \( \operatorname{trace}(-2B^T Q_{12,\omega} B_k + B_k^T Q_{k,\omega} B_k) \). Thus, by computing \( Y_{12,\omega} \) and \( Z_{12,\omega} \) separately rather than \( \mathscr{X}_{12,\omega} \), the error can be tracked, and the algorithm can be stopped if the terms \(\operatorname{trace}( -2B^T Q_{12,\omega} B_k + B_k^T Q_{k,\omega} B_k )\) stagnate, indicating that FLHNOIA's updates are no longer improving accuracy. Finally, if FLHNOIA fails to meet either of these stopping criteria, it can be terminated after the maximum allowable number of iterations, which can be set based on the permissible time and the order of the original system. By following these guidelines, the number of iterations in FLHNOIA can be controlled. In our extensive simulations, we found that, like most fixed-point iteration-based \( \mathcal{H}_2 \)-optimal MOR algorithms, FLHNOIA maintains high fidelity even when stopped after only a few iterations.
\section{Numerical Experiments}
In this section, the numerical efficiency of FLHNOIA is evaluated on three test systems. The first example uses a small illustrative system to numerically validate some of the key theoretical results presented in the paper and to facilitate the reproducibility of the mathematical results. The second example examines a single-output 1D advection-diffusion equation model taken from \cite{reiter2024h2, diaz2023interpolatory}. The third example involves a multi-output flexible space structure model from the benchmark collection of systems for testing MOR algorithms found in \cite{benner2023towards}.
\subsection{Experimental Setup}
The Gramians in BT and FLBT for LQO systems are computed using MATLAB's \textit{`lyap'} command. They are not approximated with their low-rank versions to ensure that the accuracy of BT and FLBT is not compromised by such approximations. Additionally, \( F_\omega \) in FLBT is computed exactly using MATLAB's \textit{`logm'} command to maintain the accuracy of FLBT by avoiding approximation of \( F_\omega \). Therefore, to allow a fair comparison with BT and FLBT, the first two test systems selected are of small to modest orders, ensuring that the exact computation of Gramian and the matrix logarithm \( F_\omega \) is feasible. In contrast, \( F_\omega \) in FLHNOIA is approximated as explained in Algorithm \ref{alg1} for all tests, even though exact computation of \( F_\omega \) was feasible for the examples with modest orders. This approach was adopted to demonstrate the superiority of FLHNOIA over FLBT. Furthermore, HOMORA and FLHNOIA are used to reduce a flexible space structure model of the order of one million, and the computational times are compared. The bandpass filter \( v(s) \) in Algorithm \ref{alg0} is designed using MATLAB's \textit{`butter'} command. The order \( n_v \) of the Butterworth filter is set to $16$ in the first two examples and $2$ in the third example. Both HOMORA and FLHNOIA are initialized with the same initial guess, generated arbitrarily using MATLAB's \textit{`rss'} command. The maximum number of iterations allowed in HOMORA and FLHNOIA is set to $30$ in the first two examples and $20$ in the third example. The stagnation of the eigenvalues of \( A_k \) serves as the stopping criterion for both HOMORA and FLHNOIA. The sparse-dense Sylvester equations are solved using the efficient algorithm proposed in \cite{MPIMD11-11}. The linear system of equations \( Ax = b \) in that algorithm is solved using MATLAB's backslash command `$\backslash$'. The desired frequency interval is chosen arbitrarily for demonstration purposes. All tests are performed using MATLAB R2021b on a laptop with 16 GB Random Access Memory (RAM) and a 2 GHz Intel i7 processor.
\subsection{Illustrative Example}
Consider a sixth-order LQO system defined by the following state-space representation:
\begin{align}
A &= \begin{bsmallmatrix}-9 &-29&-100&-82&-19&-2\\1&0&0&0&0&0\\0&1&0&0&0&0\\0&0&1&0&0&0\\0&0&0&1&0&0\\0&0&0&0&1&0\end{bsmallmatrix},&B &= \begin{bsmallmatrix}1\\0\\0\\0\\0\\0\end{bsmallmatrix},\nonumber\\
C &= \begin{bsmallmatrix}0&0&0&0&-1&1\end{bsmallmatrix},&M_1&=diag(\begin{smallmatrix}0.7,0.4,0.1,0.1,0.1,0.1\end{smallmatrix}).\nonumber
\end{align}
The desired frequency range for this example is set to $[5, 6]$ rad/sec. $F_\omega$, computed using (\ref{ex:66}), is the following:
\begin{align}
\begin{bsmallmatrix}
0.0551  &  0.2134  &  0.2622  &  0.1262  &  0.0248 &   0.0022\\
   -0.0011 &   0.0451  &  0.1812  &  0.1511 &  0.0351  &  0.0037\\
   -0.0019 &  -0.0178  & -0.0086 &  -0.0040 &  -0.0008 &  -0.0001\\
    0.0000 &  -0.0015  & -0.0168 &  -0.0051 &  -0.0012 &  -0.0001\\
    0.0001 &   0.0006  &  0.0003 &  -0.0105 &   0.0000 &   0.0000\\
   -0.0000 &   0.0001  &  0.0006 &   0.0002 &  -0.0106 &   0.0000
\end{bsmallmatrix}.\nonumber
\end{align}
Next, $F_\omega$ is approximated using Algorithm \ref{alg0} as $B_\omega(A, I, 5, 6, 16)$, and the following is obtained:
\begin{align}
\begin{bsmallmatrix}
 0.0555  &  0.2147  &  0.2634  &  0.1266  &  0.0249 &   0.0022\\
   -0.0011 &   0.0454 &   0.1824  &  0.1521  &  0.0354  &  0.0037\\
   -0.0019 &  -0.0179 &  -0.0086  & -0.0040  & -0.0008  & -0.0001\\
    0.0000 &  -0.0016 &  -0.0169  & -0.0052  & -0.0012  & -0.0001\\
    0.0001 &   0.0006 &   0.0003  & -0.0105  &  0.0000  &  0.0000\\
   -0.0000 &   0.0001 &   0.0006  &  0.0002  & -0.0106  &  0.0000
\end{bsmallmatrix}.\nonumber
\end{align} It can be noticed that this is a close approximation of $F_\omega$.

To initialize FLHNOIA and HOMORA, the following initial guess is used:
\begin{align}
A_k&=\begin{bsmallmatrix}-0.3918 & 0.1083  &0.0967\\
    0.1083  & -0.4423  & -0.1392\\
    0.0967  & -0.1392  & -0.2515\end{bsmallmatrix},&B_k&=\begin{bsmallmatrix}0.5022 \\ 0.3632 \\-0.0412\end{bsmallmatrix},\nonumber\\
C_k&=\begin{bsmallmatrix}0.6646  &  1.2748  &  0.3435\end{bsmallmatrix},&
M_{k,1}&=\begin{bsmallmatrix}1  &   0  &   0\\
     0   &  1  &   0\\
     0   &  0   &  1
\end{bsmallmatrix}.\nonumber
\end{align} FLHNOIA converged in $10$ iterations and produced the following ROM:
\begin{align}
A_k&=\begin{bsmallmatrix} -5.4298 & -26.3336& -102.1496\\
    0.9226  &  0.0476 &  -0.3008\\
   -0.1623  &  0.0958 &  -2.5039\end{bsmallmatrix},& B_k&=\begin{bsmallmatrix}-0.9635 \\-0.0138 \\-0.0288\end{bsmallmatrix},\nonumber\\
C_k&=\begin{bsmallmatrix}0.0012  & -0.0012  & -0.0655\end{bsmallmatrix},&
M_{k,1}&=\begin{bsmallmatrix}0.6993 &  -0.0000  &  0.0203\\
   -0.0000 &   0.3996  &  0.0003\\
    0.0203 &   0.0003  &  0.1007\end{bsmallmatrix}.\nonumber
\end{align}
The following holds for this ROM:
\begin{align}
||-P_{12,\omega}^TM_iP_{12,\omega}+P_{k,\omega}M_{k,i}P_{k,\omega}||_2&=5.0323\times 10^{-8},\nonumber\\
||-(Y_{12,\omega}+2Z_{12,\omega})^TB+(Y_{k,\omega}+2Z_{k,\omega})B_k||_2&=7.0419\times 10^{-7},\nonumber\\
||-CP_{12,\omega}+C_kP_{k,\omega}||_2&=2.9711\times 10^{-9}.\nonumber
\end{align}It can be noted that FLHNOIA satisfied the optimality conditions (\ref{op2})-(\ref{op4}), and the use of a practical bandpass filter has no meaningful impact on the satisfaction of the optimality conditions.

Next, a third-order ROM is generated using BT, HOMORA, and FLBT. HOMORA converged in 5 iterations. Figures \ref{figure1} and \ref{figure2} display the relative error on a logarithmic scale within the specified frequency range of $5$ to $6$ rad/sec. It can be noted that FLBT and FLHNOIA exhibit superior accuracy within the desired frequency range. Again, it can be noted that the use of a practical bandpass filter has no impact on accuracy as well. FLHNOIA compares well with FLBT, which is implicitly using an ideal bandpass filter.
\begin{figure}[ht]
    \centering
    \begin{subfigure}{0.48\textwidth}
        \centering
        \includegraphics[width=\linewidth]{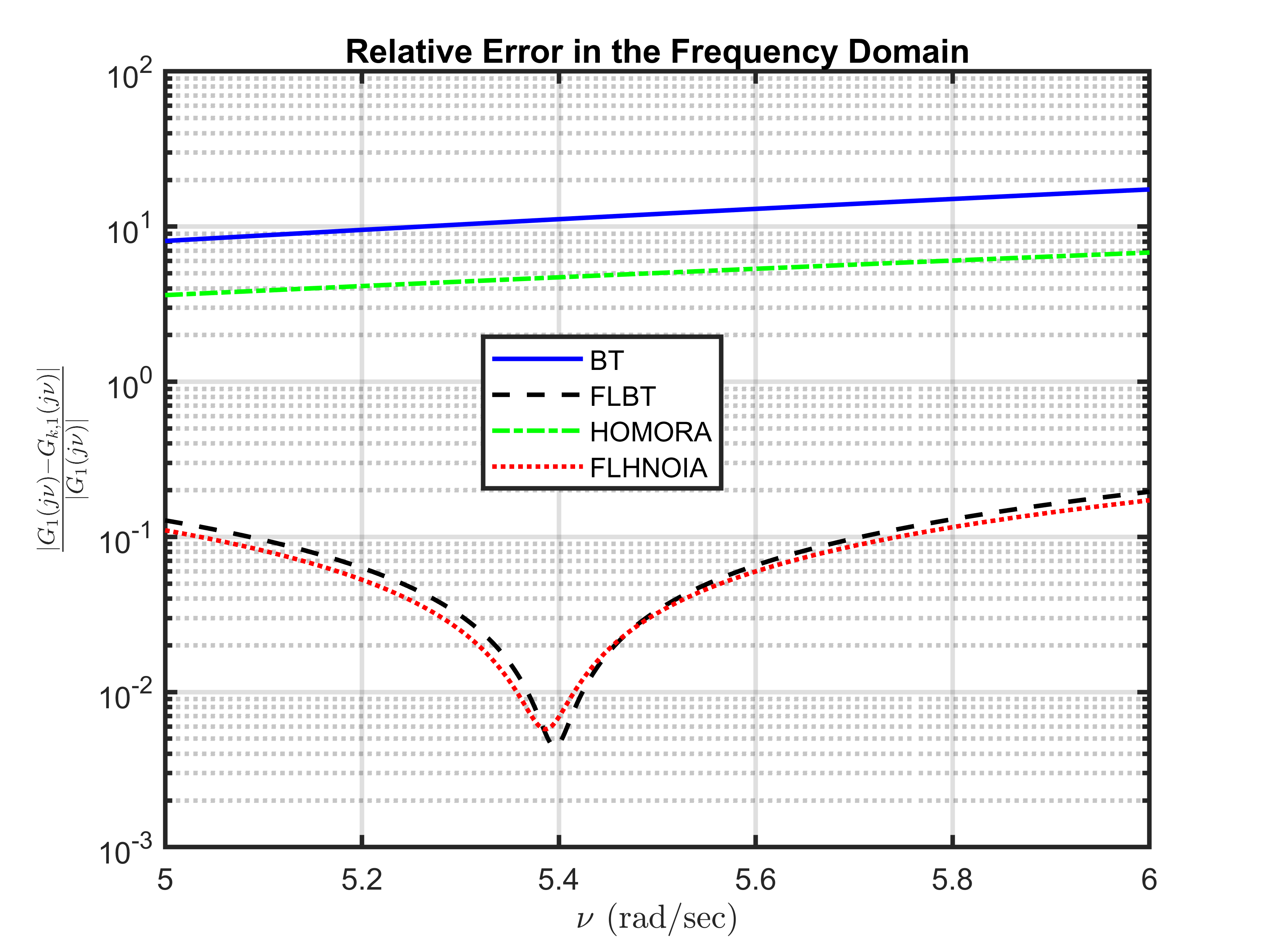}
        \caption{Relative Error $\frac{|G_1(j\nu)-G_{k,1}(j\nu)|}{|G_1(j\nu)|}$}
        \label{figure1}
    \end{subfigure}
    \hfill
    \begin{subfigure}{0.48\textwidth}
        \centering
        \includegraphics[width=\linewidth]{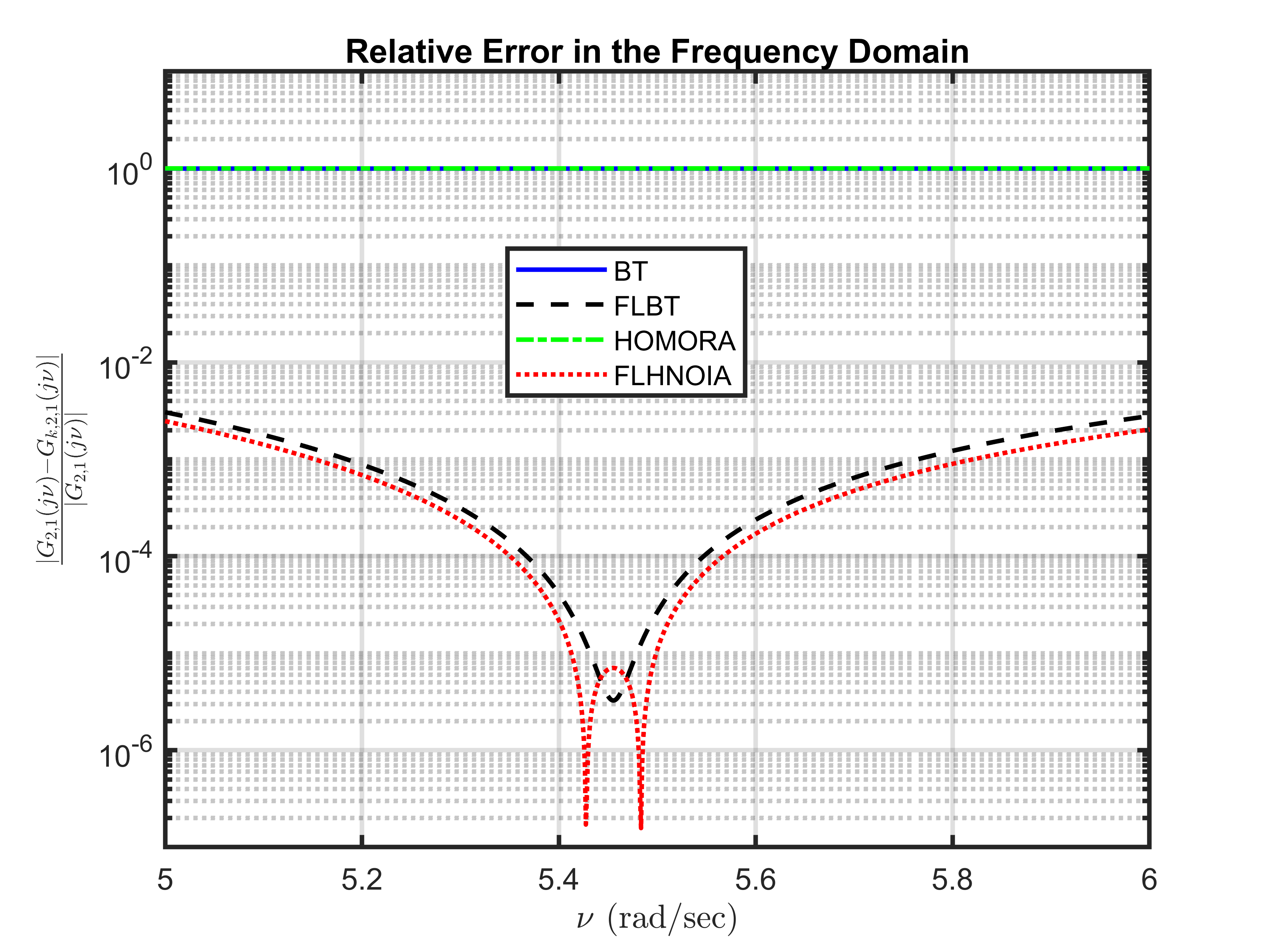}
        \caption{Relative Error $\frac{|G_{2,1}(j\nu)-G_{k,2,1}(j\nu)|}{|G_{2,1}(j\nu)|}$}
        \label{figure2}
    \end{subfigure}
    \caption{Relative Error Comparison within $[5,6]$ rad/sec}
    \label{fig1}
\end{figure}
\subsection{1D Advection-diffusion Equation}
This test system is taken from \cite{reiter2024h2, diaz2023interpolatory}. The model is briefly described as follows: The governing equations of the advection–diffusion process, characterized by diffusion \(\alpha > 0\) and advection \(\beta \geq 0\), are given by:
\begin{align} \frac{\partial}{\partial t}v(t,x)-\alpha\frac{\partial^2}{\partial x^2}v(t,x)+\beta\frac{\partial}{\partial x}v(t,x)&=0,\nonumber\\ v(t,0)=u_0(t), \hspace*{0.5cm}\alpha\frac{\partial}{\partial x}v(t,1)&=u_1(t),\nonumber\\ v(0,x)&=0,\nonumber \end{align}
where \(x \in (0,1)\), \(t \in (0,T)\), and the inputs \(u_1, u_2 \in \mathcal{L}_2(0,T)\). The output is defined by the following cost function:
\begin{align}
y(t,x) = \frac{1}{2} \int_{0}^{1} |v(t,x) - 1|^2 dx.\nonumber
\end{align}
This partial differential equation (PDE)-based model is discretized using $300$ spatial grid points, resulting in an LQO model of the form (\ref{steq:1}) with \(n = 300\), \(m = 2\), and \(p = 1\). For further details, refer to \cite{reiter2024h2, diaz2023interpolatory}.

The desired frequency interval in this example is set to \([10, 12]\) rad/sec. Given that the order of this test system is modest, \( F_\omega \) can be approximated using Algorithm \ref{alg0} as \( B_\omega(A, I, 10, 12, 16) \). The matrix \( \hat{P}_v \in \mathbb{R}^{300 \times 4800} \) can have a maximum rank of $300$. The singular values of \( \hat{P}_v \) are computed and the decay in these singular values is plotted in Figure \ref{fig3}. It can be observed from Figure \ref{fig3} that \( \hat{P}_v \) is numerically low-rank. By truncating the singular value decomposition, reduced-rank approximations of \( \hat{P}_v \) are generated and used to approximate \( F_\omega \) as \( \hat{F}_\omega \). The decay in relative error \( \frac{|\hat{F}_\omega - F_\omega|_2}{|F_\omega|_2} \) is also plotted in Figure \ref{fig3} alongside the decay in singular values of \( \hat{P}_v \). It can be seen that beyond rank $93$, there is no improvement in the accuracy of approximating \( F_\omega \). Thus, the Sylvester equation (\ref{Bw}) can be solved using low-rank ADI methods, eliminating the need to approximate \( F_\omega M_i P_{12,\omega} M_{k,i} \) in every iteration of FLHNOIA.
\begin{figure}[!h]
  \centering
  \includegraphics[width=8cm]{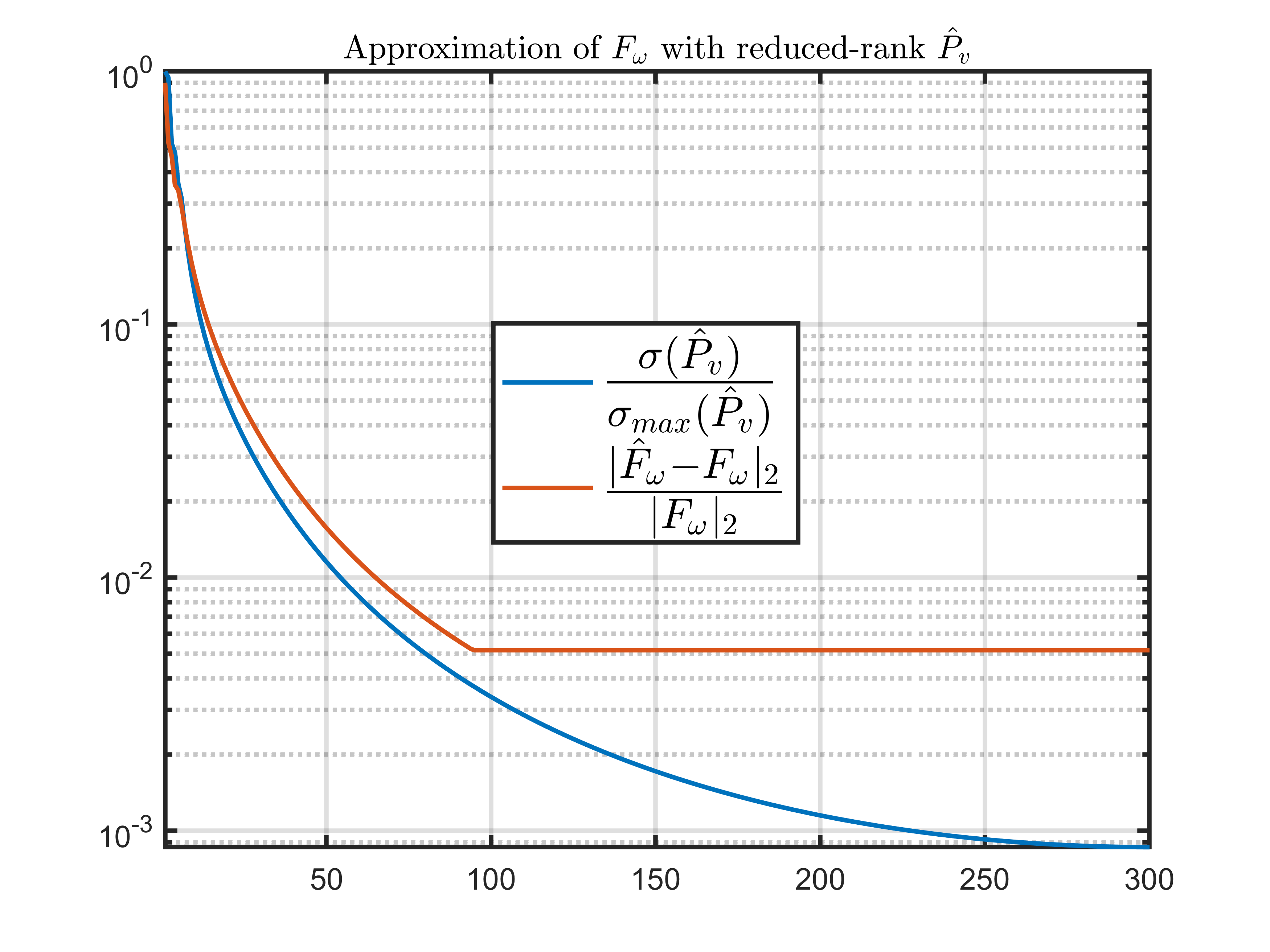}
  \caption{Decay in Singular values of $\hat{P}_v$ and Relative Error $\frac{|\hat{F}_\omega-F_\omega|_2}{|F_\omega|_2}$}\label{fig3}
\end{figure}

ROMs of orders ranging from 1 to 10 are constructed using BT, FLBT, HOMORA, and FLHNOIA. The relative errors, \(\frac{||G-G_k||_{\mathcal{H}_{2,\omega}}}{||G||_{\mathcal{H}_{2,\omega}}}\), are plotted in Figure \ref{fig4}. It is evident from the figure that FLHNOIA provides the best accuracy in terms of \(\mathcal{H}_{2,\omega}\) norm.
\begin{figure}[!h]
  \centering
  \includegraphics[width=8cm]{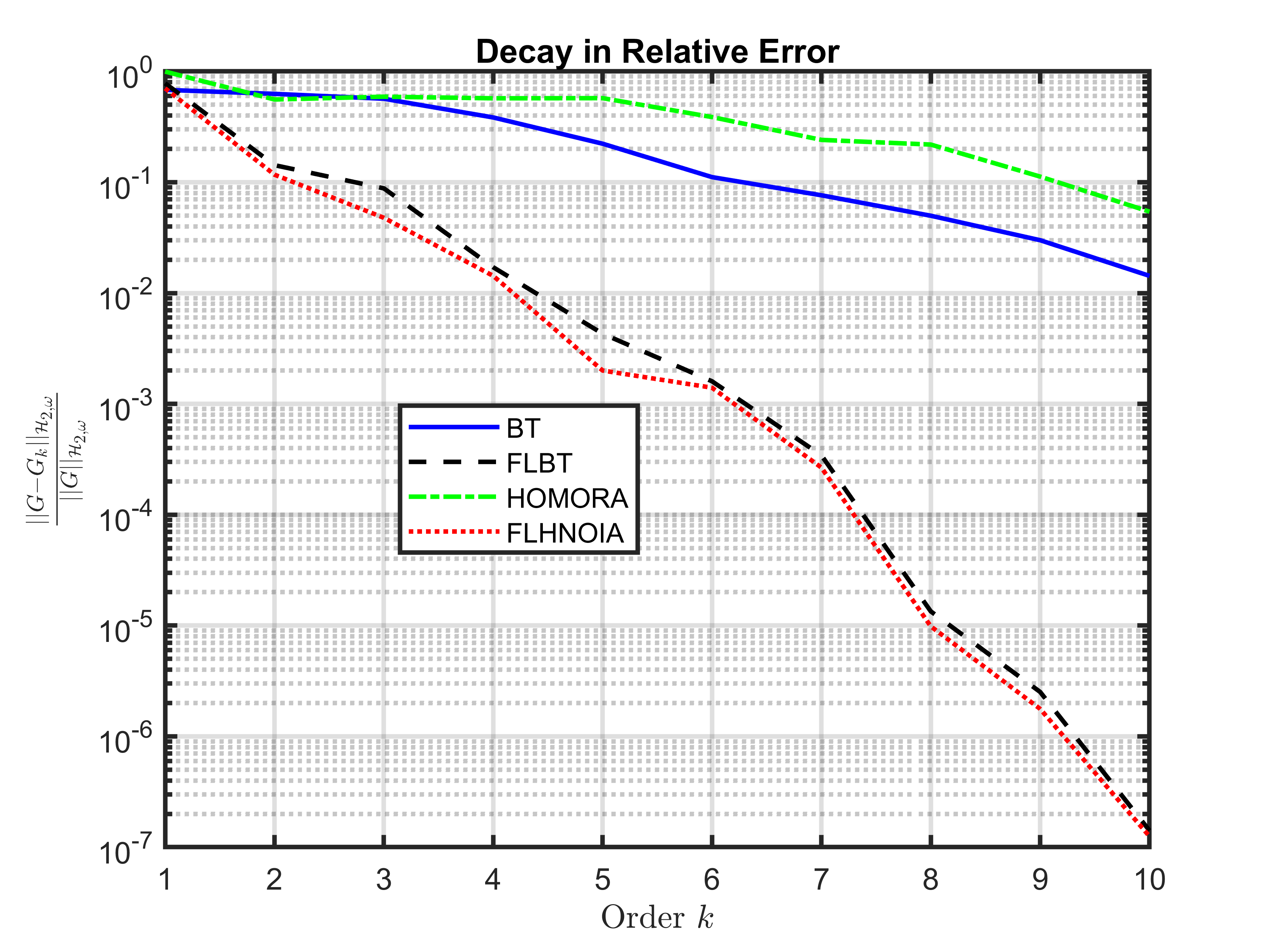}
  \caption{Relative Error $\frac{||G-G_k||_{\mathcal{H}_{2,\omega}}}{||G||_{\mathcal{H}_{2,\omega}}}$}\label{fig4}
\end{figure}
Additionally, the relative errors \(\frac{|G_1(j\nu) - G_{k,1}(j\nu)|}{|G_1(j\nu)|}\) and \(\frac{|G_{2,1}(j\nu) - G_{k,2,1}(j\nu)|}{|G_{2,1}(j\nu)|}\) within the desired frequency interval \([10,12]\) rad/sec for $10^{th}$-order reduced models are plotted in Figures \ref{figure4} and \ref{figure5}. It is evident from the figures that FLBT and FLHNOIA demonstrate superior accuracy compared to BT and HOMORA within the specified frequency range.
\begin{figure}[ht]
    \centering
    \begin{subfigure}{0.48\textwidth}
        \centering
        \includegraphics[width=\linewidth]{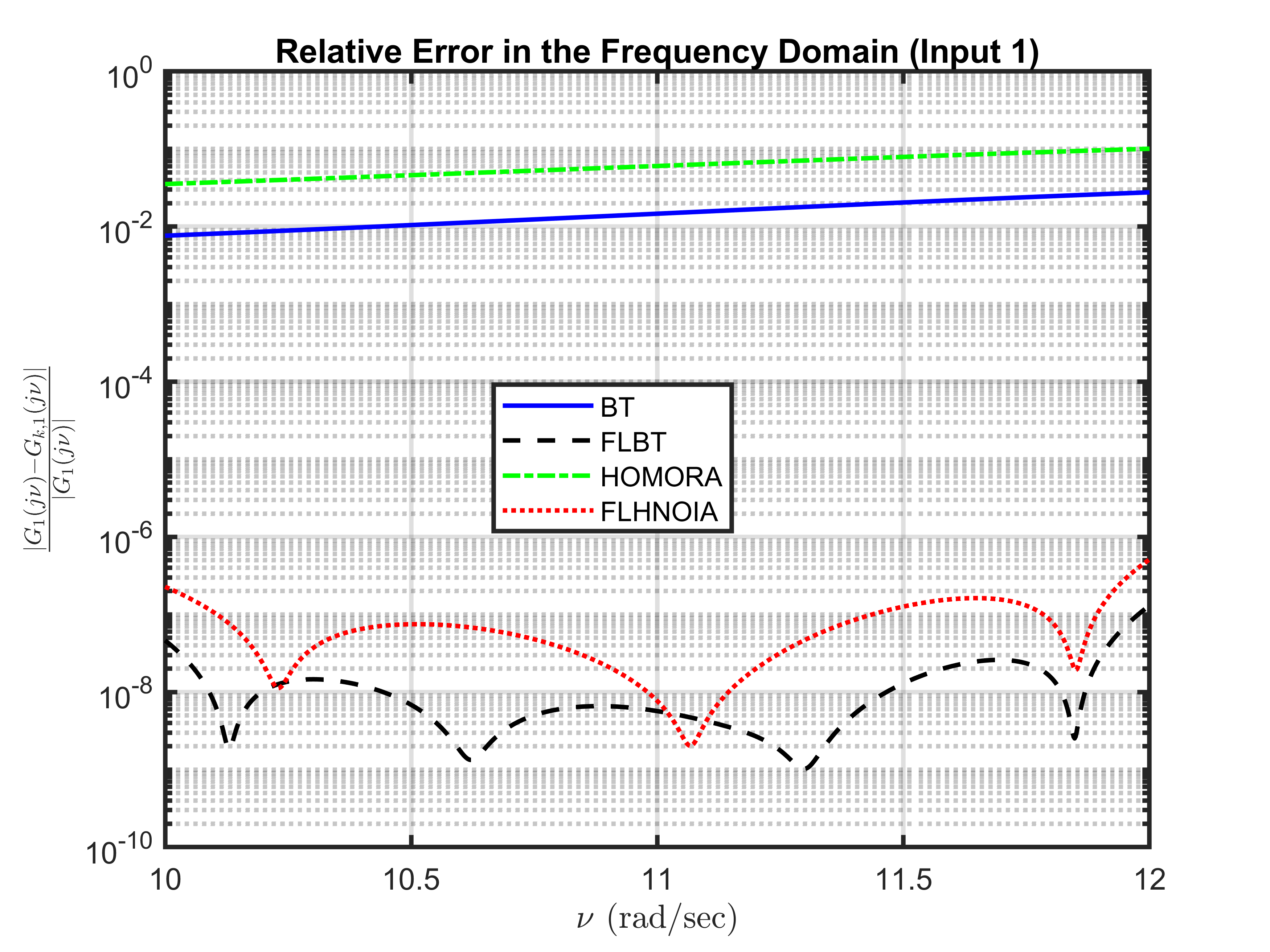}
        \caption{Relative Error $\frac{|G_1(j\nu)-G_{k,1}(j\nu)|}{|G_1(j\nu)|}$ (Input 1)}
        \label{fig5}
    \end{subfigure}
    \hfill
    \begin{subfigure}{0.48\textwidth}
        \centering
        \includegraphics[width=\linewidth]{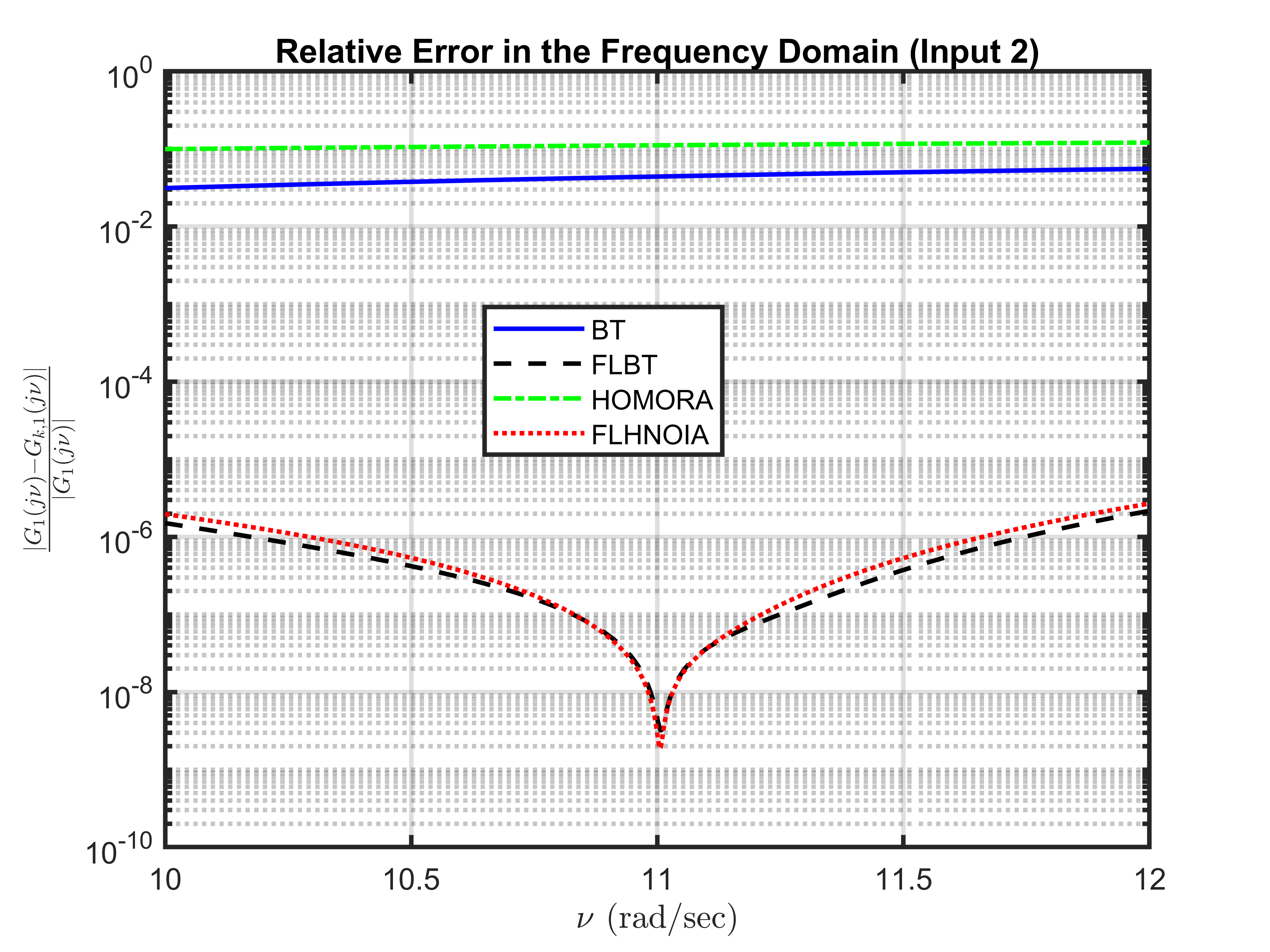}
        \caption{$\frac{|G_1(j\nu)-G_{k,1}(j\nu)|}{|G_1(j\nu)|}$ (Input 2)}
        \label{fig6}
    \end{subfigure}
    \caption{Relative Error $\frac{|G_1(j\nu)-G_{k,1}(j\nu)|}{|G_1(j\nu)|}$ Comparison within $[10,12]$ rad/sec}
    \label{figure4}
\end{figure}
\begin{figure}[ht]
    \centering
    \begin{subfigure}{0.48\textwidth}
        \centering
        \includegraphics[width=\linewidth]{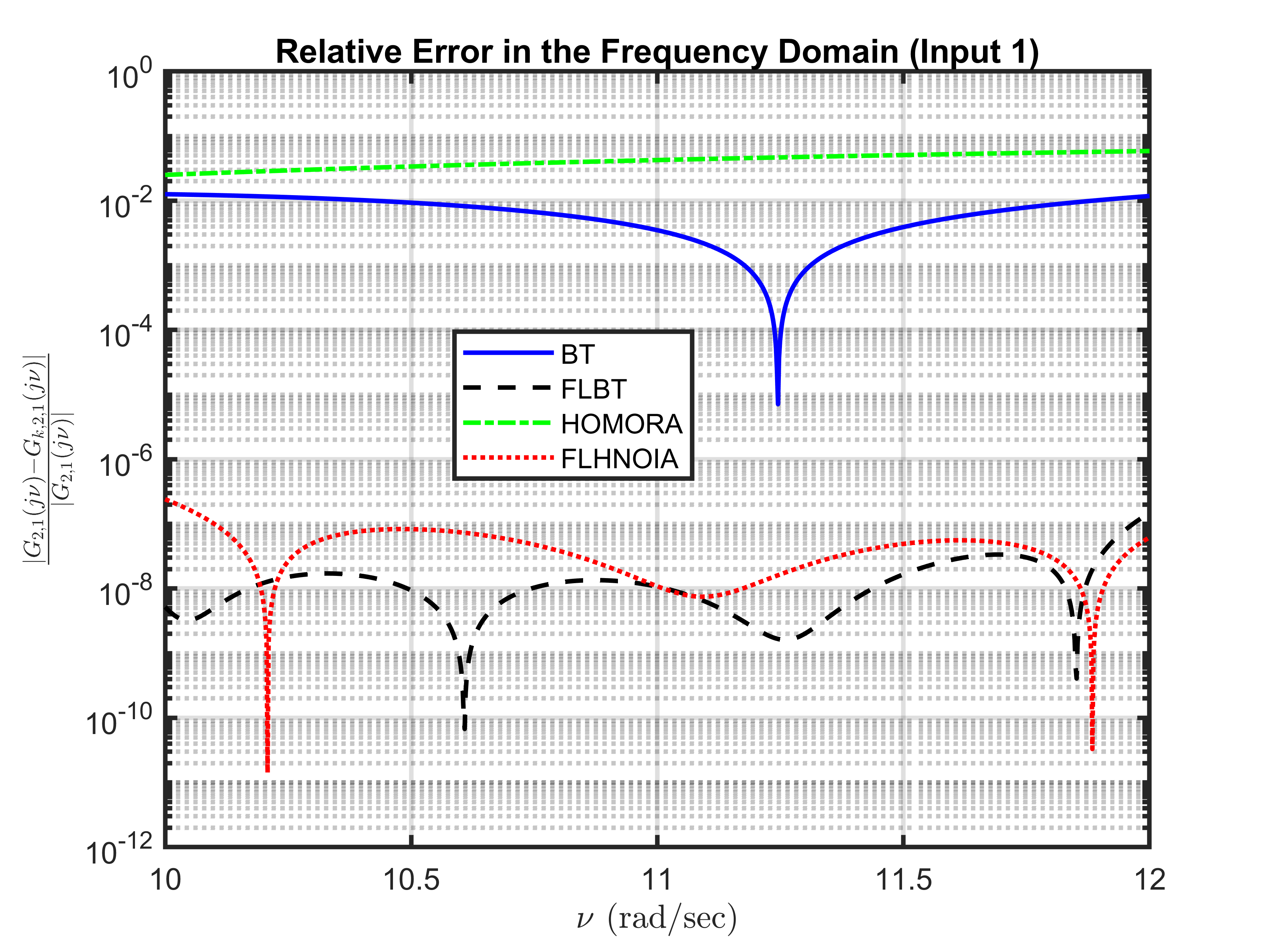}
        \caption{Relative Error $\frac{|G_2(j\nu)-G_{k,2}(j\nu)|}{|G_2(j\nu)|}$ (Input 1)}
        \label{fig5}
    \end{subfigure}
    \hfill
    \begin{subfigure}{0.48\textwidth}
        \centering
        \includegraphics[width=\linewidth]{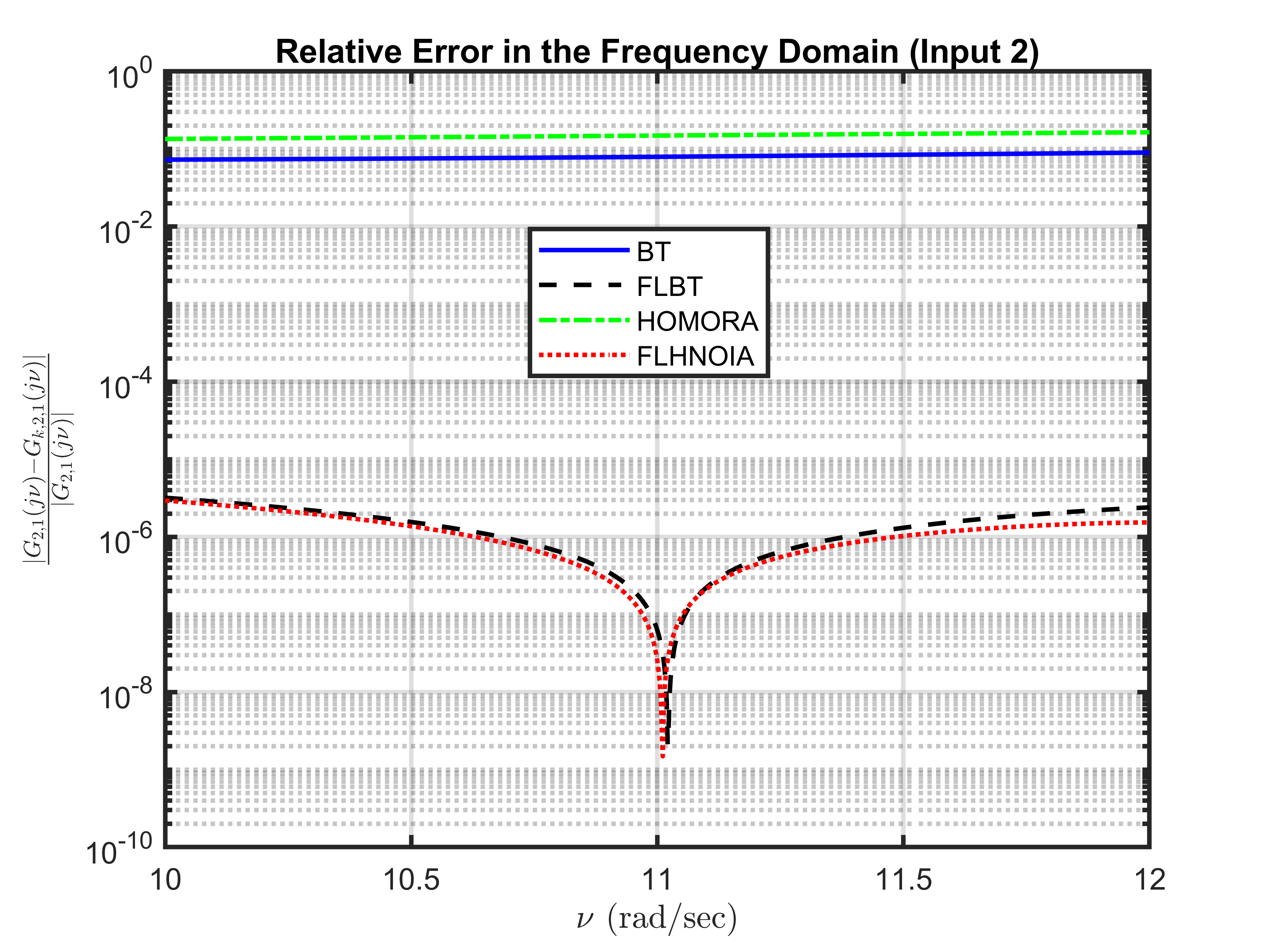}
        \caption{$\frac{|G_2(j\nu)-G_{k,2}(j\nu)|}{|G_2(j\nu)|}$ (Input 2)}
        \label{fig5}
    \end{subfigure}
    \caption{Relative Error $\frac{|G_2(j\nu)-G_{k,2}(j\nu)|}{|G_2(j\nu)|}$ Comparison within $[10,12]$ rad/sec}
    \label{figure5}
\end{figure}
\subsection{Flexible Space Structure}
The flexible space structure (FSS) benchmark \cite{benner2023towards} is a procedural modal model used for simulating structural dynamics, particularly for truss structures in space environments, such as the COFS-1 (Control of Flexible Structures) mast flight experiment \cite{horner1986cofs,horta1986analysis}. In modal form, the FSS model for $\frac{n}{2}$ modes, $m$ actuators, and $p$ sensors is described by the following second-order system:
\begin{align}
\ddot{\nu}(t)&=(2\zeta\circ \omega)\circ\dot{\nu}(t)+(\omega\circ\omega)\circ\nu(t)+Bu(t),\nonumber\\
y(t)&=C_r\dot{\nu}(t)+C_d\nu(t),\nonumber
\end{align}where \(\zeta\) and \(\omega\) represent the damping ratio and natural frequency of the modes, respectively, and \(\circ\) denotes the Hadamard product. This second-order model is converted to a first-order state-space form to obtain an $n^{th}$-order model with $m$ inputs and $p$ outputs; see \cite{benner2023towards} for details.

For this example, we selected $500,000$ modes, one actuator, and two sensors, resulting in a FSS model with 1 million states, using MATLAB codes provided in \cite{benner2023towards}. We added two quadratic outputs, \(y_{2,1}(t)\) and \(y_{2,2}(t)\), to this standard LTI system to quantify the energy of $1,000$ and $2,000$ randomly selected states \(x(t) = \begin{bmatrix}x_1(t) & \cdots & x_n(t)\end{bmatrix}^T\), respectively:
\begin{align}
y_{2,1}(t)&=\frac{1}{2}\sum_{i=1}^{n}m_{1,i}\big(x_i(t)\big)^2&\textnormal{and}&&y_{2,2}(t)&=\frac{1}{2}\sum_{i=1}^{n}m_{2,i}\big(x_i(t)\big)^2,\nonumber
\end{align}where \(m_{1,i}\) and \(m_{2,i}\) are either $1$ or $0$. The matrices \(M_1 = 0.5 \, diag(m_{1,1}, \cdots, m_{1,n})\) and \(M_2 = 0.5 \, diag(m_{2,1}, \cdots, m_{2,n})\) have $1,000$ and $2,000$ nonzero diagonal elements, respectively, selected randomly using MATLAB's \textit{`rand'} command. The state-space matrices of the resulting LQO system have the following dimensions: \(n = 1,000,000\), \(m = 1\), and \(p = 2\). The desired frequency interval in this example is set to \([28, 29]\) rad/sec. Due to the limited memory of the computer used, exact computation of the Gramians in BT and FLBT is not feasible for this example, nor is the exact computation of \(F_\omega\). We did not pursue low-rank approximations of the Gramians or Krylov-subspace-based approximations of \(F_\omega\). Instead, we compared FLHNOIA with HOMORA, as both algorithms can handle systems of this order directly. $30^{th}$-order reduced models were constructed using FLHNOIA and HOMORA. The computational time to generate the $30^{th}$-order reduced model using HOMORA was $12.6246$ minutes, while FLHNOIA took $21.2882$ minutes. There are four frequency domain error plots. For economy of space, only the relative errors \(\frac{|G_1(j\nu) - G_{k,1}(j\nu)|}{|G_1(j\nu)|}\) for output 1 are shown in Figure \ref{fig9}. It can be seen that FLHNOIA provides superior accuracy compared to HOMORA within the desired frequency interval \([28, 29]\) rad/sec.
\begin{figure}[!h]
  \centering
  \includegraphics[width=8cm]{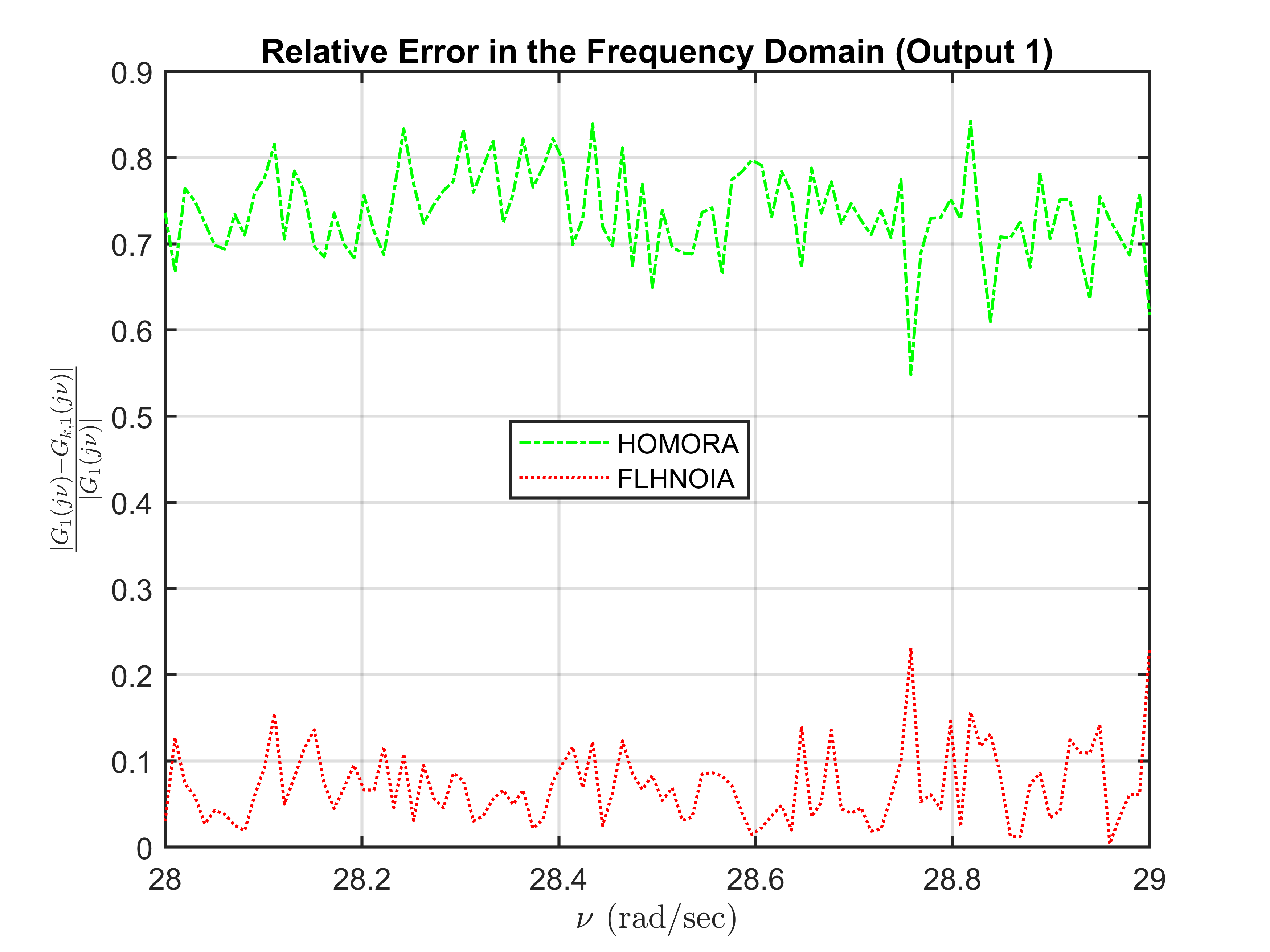}
  \caption{Relative Error $\frac{|G_1(j\nu)-G_{k,1}(j\nu)|}{|G_1(j\nu)|}$ (Output 1) within $[28,29]$ rad/sec}\label{fig9}
\end{figure}
\section{Conclusion}
This research addresses the problem of $\mathcal{H}_2$-optimal MOR within a specified finite frequency range. To measure the output strength within this range, we introduce the frequency-limited $\mathcal{H}_2$ norm for LQO systems. We derive the necessary conditions for achieving local optima of the squared frequency-limited $\mathcal{H}_2$ norm of the error and compare these conditions to those of the standard, unconstrained $\mathcal{H}_2$-optimal MOR problem. The study highlights the limitations of the Petrov-Galerkin projection method in fully satisfying all optimality conditions in the frequency-limited context. Consequently, we propose a Petrov-Galerkin projection algorithm that approximately satisfies these optimality conditions. The paper also discusses the computationally efficient implementation of the proposed algorithm. Unlike FLBT, the proposed algorithm eliminates the need for Gramian computation, making it significantly more efficient for large-scale systems. Numerical experiments are conducted to validate the theoretical results and demonstrate the algorithm's effectiveness in achieving high accuracy within the specified frequency range. The numerical results support the theory developed in the paper and show that the proposed algorithm constructs a ROM efficiently, maintaining high fidelity within the desired frequency range.
\section*{Acknowledgement}
This work is supported by the National Natural Science Foundation of China under Grants No. 62350410484 and 62273059, and in part by the High-end Foreign Expert Program No. G2023027005L granted by the State Administration of Foreign Experts Affairs (SAFEA).
\section*{Appendix}
In this appendix, we present the proof of Theorem \ref{th1}. Throughout the proof, the following properties of the trace operation are utilized repeatedly:
\begin{enumerate}
  \item Trace of Hermitian: $\operatorname{trace}(XYZ)=\operatorname{trace}(Z^*Y^*Z^*)$.
  \item Circular permutation in Trace: $\operatorname{trace}(XYZ)=\operatorname{trace}(ZXY)=\operatorname{trace}(YZX)$.
  \item Trace of addition: $\operatorname{trace}(X+Y+Z)=\operatorname{trace}(X)+\operatorname{trace}(Y)+\operatorname{trace}(Z)$;
\end{enumerate}cf. \cite{petersen2008matrix}.

Let us define the cost function  $J$ as the component of $||G-G_k||_{H_{2,\omega}}^2$ that depends on the ROM, expressed as:
\begin{align}
J=\operatorname{trace}(-2B^TQ_{12,\omega}B_k+B_k^TQ_{k,\omega}B_k).\nonumber
\end{align}When a small first-order perturbation $\Delta_{A_k}$ is added to $A_k$, $J$ changes to $J+\Delta_{J}^{A_k}$. This causes $Q_{12,\omega}$ and $Q_{k,\omega}$ to perturb to $Q_{12,\omega}+\Delta_{Q_{12,\omega}}^{A_k}$ and $Q_{k,\omega}+\Delta_{Q_{k,\omega}}^{A_k}$, respectively. Consequently, the first-order terms of $\Delta_{J}^{A_k}$ are given by:
\begin{align}
\Delta_{J}^{A_k}=\operatorname{trace}(2B^T\Delta_{Q_{12,\omega}}^{A_k}B_k+B_k^T\Delta_{Q_{k,\omega}}^{A_k}B_k).\nonumber
\end{align}
Furthermore, it is evident from (\ref{eq:30}) and (\ref{eq:31}) that $\Delta_{Q_{12,\omega}}^{A_k}$ and $\Delta_{Q_{k,\omega}}^{A_k}$ satisfy the following Lyapunov equations:
\begin{align}
&A^T\Delta_{Q_{12,\omega}}^{A_k}+\Delta_{Q_{12,\omega}}^{A_k}A_k+Q_{12,\omega}\Delta_{A_k}+C^TC_k\Delta_{F_{k,\omega}}^{A_k}+\sum_{i=1}^{p}\big(M_iP_{12,\omega}M_{k,i}\Delta_{F_{k,\omega}}^{A_k}\nonumber\\
&\hspace*{4.5cm}+F_\omega^*M_i\Delta_{P_{12,\omega}}^{A_k}M_{k,i}+M_i\Delta_{P_{12,\omega}}^{A_k}M_{k,i}F_{k,\omega}\big)=0,\nonumber\\
&\hspace*{2.85cm}A\Delta_{P_{12,\omega}}^{A_k}+\Delta_{P_{12,\omega}}^{A_k}A_k^T+P_{12,\omega}(\Delta_{A_k})^T+BB_k^T(\Delta_{F_{k,\omega}}^{A_k})^*=0,\nonumber\\
&A_k^T\Delta_{Q_{k,\omega}}^{A_k}+\Delta_{Q_{k,\omega}}^{A_k}A_k+(\Delta_{A_k})^TQ_{k,\omega}+Q_{k,\omega}\Delta_{A_k}\nonumber\\
&\hspace*{0.65cm}+(\Delta_{F_{k,\omega}}^{A_k})^*C_k^TC_k+C_k^TC_k\Delta_{F_{k,\omega}}^{A_k}+\sum_{i=1}^{p}\big((\Delta_{F_{k,\omega}}^{A_k})^*M_{k,i}P_{k,\omega}M_{k,i}\nonumber\\
&\hspace*{0.65cm}+M_{k,i}P_{k,\omega}M_{k,i}\Delta_{F_{k,\omega}}^{A_k}+F_{k,\omega}^*M_{k,i}\Delta_{P_{k,\omega}}^{A_k}M_{k,i}+M_{k,i}\Delta_{P_{k,\omega}}^{A_k}M_{k,i}F_{k,\omega}\big)=0,\nonumber\\
&A_k\Delta_{P_{k,\omega}}^{A_k}+\Delta_{P_{k,\omega}}^{A_k}A_k^T+\Delta_{A_k}P_{k,\omega}+P_{k,\omega}(\Delta_{A_k})^T\nonumber\\
&\hspace*{6cm}+\Delta_{F_{k,\omega}}^{A_k}B_kB_k^T+B_kB_k^T(\Delta_{F_{k,\omega}}^{A_k})^*=0,\nonumber
\end{align}
where
\begin{align}
\Delta_{F_{k,\omega}}^{A_k}=\frac{j}{\pi}\mathcal{L}(-j\nu I-A_k,\Delta_{A_k})+o(||\Delta_{A_k}||);\nonumber
\end{align}cf. \cite{higham2008functions}. Since we are only concerned with first-order perturbations, the term $o(||\Delta_{A_k}||)$ will be omitted for the rest of the proof.
Now,
\begin{align}
&\operatorname{trace}\Big(BB_k^T(\Delta_{Q_{12,\omega}}^{A_k})^*\Big)\nonumber\\
&=\operatorname{trace}\Big((-AP_{12}-P_{12}A_k^T)(\Delta_{Q_{12,\omega}}^{A_k})^*\Big)\nonumber\\
&=\operatorname{trace}\Big(-AP_{12}(\Delta_{Q_{12,\omega}}^{A_k})^*-P_{12}A_k^T(\Delta_{Q_{12,\omega}}^{A_k})^*\Big)\nonumber\\
&=\operatorname{trace}\Big(P_{12}^T(-A^T\Delta_{Q_{12,\omega}}^{A_k}-A_k\Delta_{Q_{12,\omega}}^{A_k})\Big)\nonumber\\
&=\operatorname{trace}\Big(P_{12}^TQ_{12,\omega}\Delta_{A_k}+P_{12}^TC^TC_k\Delta_{F_{k,\omega}}^{A_k}+\sum_{i=1}^{p}\big(P_{12}^TM_iP_{12,\omega}M_{k,i}\Delta_{F_{k,\omega}}^{A_k}\nonumber\\
&\hspace*{2cm}+P_{12}^TF_\omega^TM_i\Delta_{P_{12,\omega}}^{A_k}M_{k,i}+F_{k,\omega}P_{12}^TM_i\Delta_{P_{12,\omega}}^{A_k}M_{k,i}\big)\Big)\nonumber\\
&=\operatorname{trace}\Big(Q_{12,\omega}^*P_{12}\Delta_{A_k}^T+P_{12}^TC^TC_k\Delta_{F_{k,\omega}}^{A_k}+\sum_{i=1}^{p}\big(P_{12}^TM_iP_{12,\omega}M_{k,i}\Delta_{F_{k,\omega}}^{A_k}\nonumber\\
&\hspace*{2cm}+P_{12}^TF_\omega^*M_i\Delta_{P_{12,\omega}}^{A_k}M_{k,i}+F_{k,\omega}P_{12}^TM_i\Delta_{P_{12,\omega}}^{A_k}M_{k,i}\big)\Big).\nonumber
\end{align}
Similarly, note that:
\begin{align}
&\operatorname{trace}(B_kB_k^T\Delta_{Q_{k,\omega}}^{A_k})\nonumber\\
&=\operatorname{trace}\Big(\big(-A_kP_k-P_kA_k^T\big)\Delta_{Q_{k,\omega}}^{A_k}\Big)\nonumber\\
&=\operatorname{trace}\Big(P_k\big(-A_k^T\Delta_{Q_{k,\omega}}^{A_k}-\Delta_{Q_{k,\omega}}^{A_k}A_k\big)\Big)\nonumber\\
&=\operatorname{trace}\Bigg(P_k\Big((\Delta_{A_k})^TQ_{k,\omega}+Q_{k,\omega}\Delta_{A_k}+(\Delta_{F_{k,\omega}}^{A_k})^*C_k^TC_k+C_k^TC_k\Delta_{F_{k,\omega}}^{A_k}\nonumber\\
&\hspace*{2cm}+\sum_{i=1}^{p}\big((\Delta_{F_{k,\omega}}^{A_k})^*M_{k,i}P_{k,\omega}M_{k,i}+M_{k,i}P_{k,\omega}M_{k,i}\Delta_{F_{k,\omega}}^{A_k}\nonumber\\
&\hspace*{2cm}+F_{k,\omega}^*M_{k,i}\Delta_{P_{k,\omega}}^{A_k}M_{k,i}+M_{k,i}\Delta_{P_{k,\omega}}^{A_k}M_{k,i}F_{k,\omega}\big)\Big)\Bigg)\nonumber\\
&=\operatorname{trace}\Big(2Q_{k,\omega}P_k(\Delta_{A_k})^T+2(\Delta_{F_{k,\omega}}^{A_k})^*C_k^TC_kP_k\nonumber\\
&\hspace*{2cm}+\sum_{i=1}^{p}\big(2(\Delta_{F_{k,\omega}}^{A_k})^*M_{k,i}P_{k,\omega}M_{k,i}P_k+M_{k,i}F_{k,\omega}P_kM_k\Delta_{P_{k,\omega}}^{A_k}\nonumber\\
&\hspace*{2cm}+M_{k,i}P_kF_{k,\omega}^*M_{k,i}\Delta_{P_{k,\omega}}^{A_k}\big)\Big).\nonumber
\end{align}
Therefore:
\begin{align}
&\Delta_{J}^{A_k}=\operatorname{trace}\Big(-2(Q_{12,\omega})^*P_{12}(\Delta_{A_k})^T+2Q_{k,\omega}P_k(\Delta_{A_k})^T\nonumber\\
&\hspace*{2cm}-2(\Delta_{F_{k,\omega}}^{A_k})^*C_k^TCP_{12}+2(\Delta_{F_{k,\omega}}^{A_k})^*C_k^TC_kP_k\nonumber\\
&\hspace*{2cm}+\sum_{i=1}^{p}\big(-2M_iF_\omega P_{12}M_{k,\omega}(\Delta_{P_{12,\omega}}^{A_k})^*-2M_iP_{12}F_{k,\omega}^*M_{k,i}(\Delta_{P_{12,\omega}}^{A_k})^*\nonumber\\
&\hspace*{2cm}+M_{k,i}F_{k,\omega}P_kM_{k,i}\Delta_{P_{k,\omega}}^{A_k}+M_{k,i}P_kF_{k,\omega}^*M_{k,i}\Delta_{P_{k,\omega}}^{A_k}\nonumber\\
&\hspace*{2cm}-2(\Delta_{F_{k,\omega}}^{A_k})^*M_{k,i}(P_{12,\omega})^*M_iP_{12}+2(\Delta_{F_{k,\omega}}^{A_k})^*M_{k,i}P_{k,\omega}M_{k,i}P_k\big)\Big).\nonumber
\end{align}
Since
\begin{align}
P_{12,\omega}=F_\omega P_{12}+P_{12}F_{k,\omega}^*,\nonumber\\
P_{k,\omega}=F_{k,\omega}P_k+P_kF_{k,\omega}^*,\nonumber
\end{align}
we have:
\begin{align}
&\Delta_{J}^{A_k}=\operatorname{trace}\Big(-2(Q_{12,\omega})^*P_{12}(\Delta_{A_k})^T+2Q_{k,\omega}P_k(\Delta_{A_k})^T\nonumber\\
&\hspace*{1cm}-2(\Delta_{F_{k,\omega}}^{A_k})^*C_k^TCP_{12}+2(\Delta_{F_{k,\omega}}^{A_k})^*C_k^TC_kP_k\nonumber\\
&\hspace*{1cm}+\sum_{i=1}^{p}\big(-2M_iP_{12,\omega}M_{k,i}(\Delta_{P_{12,\omega}}^{A_k})^*+M_{k,i}P_{k,\omega}M_{k,i}\Delta_{P_{k,\omega}}^{A_k}\nonumber\\
&\hspace*{1cm}-2(\Delta_{F_{k,\omega}}^{A_k})^*M_{k,i}(P_{12,\omega})^*M_iP_{12}+2(\Delta_{F_{k,\omega}}^{A_k})^*M_{k,i}P_{k,\omega}M_{k,i}P_k\big)\Big);\nonumber
\end{align}cf. \cite{zulfiqar2022adaptive}.
Note that
\begin{align}
&\operatorname{trace}\Big(\sum_{i=1}^{p}M_iP_{12,\omega}M_{k,i}(\Delta_{P_{12,\omega}}^{A_k})^*\Big)\nonumber\\
&=\operatorname{trace}\Big(\big(-A^T\bar{Z}_{12}-\bar{Z}_{12}A_k\big)(\Delta_{P_{12,\omega}}^{A_k})^*\Big)\nonumber\\
&=\operatorname{trace}\Big(\big(-A\Delta_{P_{12,\omega}}^{A_k}-\Delta_{P_{12,\omega}}^{A_k}A_k^T\big)\bar{Z}_{12}^*\Big)\nonumber\\
&=\operatorname{trace}\big(\bar{Z}_{12}^*P_{12,\omega}(\Delta_{A_k})^T+\bar{Z}_{12}^*BB_k^T(\Delta_{F_{k,\omega}}^{A_k})^T\big).\nonumber
\end{align}
Additionally, note that
\begin{align}
&\operatorname{trace}\Big(\sum_{i=1}^{p}M_{k,i}P_{k,\omega}M_{k,i}\Delta_{P_{k,\omega}}^{A_k}\Big)\nonumber\\
&=\operatorname{trace}\Big(\big(-A_k^T\bar{Z}_k-\bar{Z}_kA_k\big)\Delta_{P_{k,\omega}}^{A_k}\Big)\nonumber\\
&=\operatorname{trace}\Big(\big(-A_k\Delta_{P_{k,\omega}}^{A_k}-\Delta_{P_{k,\omega}}^{A_k}A_k^T\big)\bar{Z}_k\Big)\nonumber\\
&=2\operatorname{trace}\big(\bar{Z}_kP_{k,\omega}(\Delta_{A_k})^T+\bar{Z}_kB_kB_k^T(\Delta_{F_{k,\omega}}^{A_k})^T\big).\nonumber
\end{align}
Thus,
\begin{align}
\Delta_{J}^{A_k}&=\operatorname{trace}\Big(-2(Q_{12,\omega})^*P_{12}(\Delta_{A_k})^T+2Q_{k,\omega}P_k(\Delta_{A_k})^T\nonumber\\
&\hspace*{2cm}-2\bar{Z}_{12}^*P_{12,\omega}(\Delta_{A_k})^T+2\bar{Z}_kP_{k,\omega}(\Delta_{A_k})^T\nonumber\\
&\hspace*{2cm}-2B_kB^T\bar{Z}_{12}\Delta_{F_{k,\omega}}^{A_k}+2B_kB_k^T\bar{Z}_k\Delta_{F_{k,\omega}}^{A_k}\nonumber\\
&\hspace*{2cm}-2P_{12}^TC^TC_k\Delta_{F_{k,\omega}}^{A_k}+2P_kC_kC_k^T\Delta_{F_{k,\omega}}^{A_k}\nonumber\\
&\hspace*{2cm}-2\sum_{i=1}^{p}P_{12}^TM_iP_{12,\omega}M_{k,i}\Delta_{F_{k,\omega}}^{A_k}\nonumber\\
&\hspace*{2cm}+2\sum_{i=1}^{p}P_kM_{k,i}P_{k,\omega}M_{k,i}\Delta_{F_{k,\omega}}^{A_k}\Big)\nonumber\\
&=\operatorname{trace}\Big(-2(Q_{12,\omega})^*P_{12}(\Delta_{A_k})^T+2Q_{k,\omega}P_k(\Delta_{A_k})^T\nonumber\\
&\hspace*{2cm}-2\bar{Z}_{12}^*P_{12,\omega}(\Delta_{A_k})^T+2\bar{Z}_kP_{k,\omega}(\Delta_{A_k})^T-2V\Delta_{F_{k,\omega}}^{A_k}\Big).\nonumber
\end{align}
Recall that
\begin{align}
\Delta_{F_{k,\omega}}^{A_k}=\frac{j}{2\pi}\mathcal{L}(-A_k-j\nu I,-\Delta_{A_k}).\nonumber
\end{align}
By exchanging the trace and integral operations, we obtain
\begin{align}
\operatorname{trace}(V\Delta_{F_{k,\omega}}^{A_k})=-\operatorname{trace}(W\Delta_{A_k});\nonumber
\end{align}cf. \cite{petersson2013nonlinear}. Hence,
\begin{align}
\Delta_{J}^{A_k}=2\operatorname{trace}\Big(\big(-Q_{12,\omega}^*P_{12}+Q_{k,\omega}P_k-\bar{Z}_{12}^*P_{12,\omega}+\bar{Z}_kP_{k,\omega}+W^*\big)(\Delta_{A_k})^T\Big).\nonumber
\end{align}
Therefore, the gradient of $J$ with respect of $A_k$ is given by
\begin{align}
\nabla_{J}^{A_k}=2\big(-Q_{12,\omega}^*P_{12}+Q_{k,\omega}P_k-\bar{Z}_{12}^*P_{12,\omega}+\bar{Z}_kP_{k,\omega}+W^*\big);\nonumber
\end{align}cf. \cite{higham2008functions}.
Consequently,
\begin{align}
-Q_{12,\omega}^*P_{12}+Q_{k,\omega}P_k-\bar{Z}_{12}^*P_{12,\omega}+\bar{Z}_kP_{k,\omega}+W^*=0\label{nst79}
\end{align} is a necessary condition for a local optimum of $||G-G_k||_{H_{2,\omega}}^2$. Moreover, substituting (\ref{nst46})-(\ref{nst49}) into (\ref{nst79}), it simplifies to
\begin{align}
-Q_{12,\omega}^*P_{12,\omega}+Q_{k,\omega}P_{k,\omega}-Z_{12,\omega}^*P_{12,\omega}+Z_{k,\omega}P_{k,\omega}+L_\omega=0.\nonumber
\end{align} Additionally, since $Q_{12,\omega}=Y_{12,\omega}+Z_{12,\omega}$ and $Q_{k,\omega}=Y_{k,\omega}+Z_{k,\omega}$, we arrive at
\begin{align}
-(Y_{12,\omega}+2Z_{12,\omega})^*P_{12,\omega}+(Y_{k,\omega}+2Z_{k,\omega})P_{k,\omega}+L_\omega&=0.\nonumber
\end{align}

By introducing a small first-order perturbation $\Delta_{M_{k,i}}$ to $M_{k,i}$, $J$ is perturbed to $J+\Delta_{J}^{M_{k,i}}$. Consequently, $Q_{12,\omega}$ and $Q_{k,\omega}$ are perturbed to $Q_{12,\omega}+\Delta_{Q_{12,\omega}}^{M_{k,i}}$ and $Q_{k,\omega}+\Delta_{Q_{k,\omega}}^{M_{k,i}}$, respectively. As a result, the first-order terms of $\Delta_{J}^{M_{k,i}}$ are given by
\begin{align}
\Delta_{J}^{M_{k,i}}=\operatorname{trace}(-2B^T\Delta_{Q_{12,\omega}}^{M_{k,i}}B_k+B_k^T\Delta_{Q_{k,\omega}}^{M_{k,i}}B_k).\nonumber
\end{align}
Furthermore, it can be easily observed from (\ref{eq:30}) and (\ref{eq:31}) that $\Delta_{Q_{12,\omega}}^{M_{k,i}}$ and $\Delta_{Q_{k,\omega}}^{M_{k,i}}$ satisfy the following Lyapunov equations:
\begin{align}
&A^T\Delta_{Q_{12,\omega}}^{M_{k,i}}+\Delta_{Q_{12,\omega}}^{M_{k,i}}A_k+F_\omega^*M_iP_{12,\omega}\Delta_{M_{k,i}}+M_iP_{12,\omega}\Delta_{M_{k,i}}F_{k,\omega}=0,\nonumber\\
&A_k^T\Delta_{Q_{k,\omega}}^{M_{k,i}}+\Delta_{Q_{k,\omega}}^{M_{k,i}}A_k+F_{k,\omega}^*\Delta_{M_{k,i}}P_{k,\omega}M_{k,i}+\Delta_{M_{k,i}}P_{k,\omega}M_{k,i}F_{k,\omega}\nonumber\\
&\hspace{3cm}+F_{k,\omega}^*M_{k,i}P_{k,\omega}\Delta_{M_{k,i}}+M_{k,i}P_{k,\omega}\Delta_{M_{k,i}}F_{k,\omega}=0.\nonumber
\end{align}
Observe that
\begin{align}
&\operatorname{trace}(B^T\Delta_{Q_{12,\omega}}^{M_{k,i}}B_k)\nonumber\\
&=\operatorname{trace}\big(BB_k^T(\Delta_{Q_{12,\omega}}^{M_{k,i}})^*\big)\nonumber\\
&=\operatorname{trace}\Big(\big(-AP_{12}-P_{12}A_k^T\big)(\Delta_{Q_{12,\omega}}^{M_{k,i}})^*\Big)\nonumber\\
&=\operatorname{trace}\Big(\big(-A^T\Delta_{Q_{12,\omega}}^{M_{k,i}}-\Delta_{Q_{12,\omega}}^{M_{k,i}}A_k\big)P_{12}^T\Big)\nonumber\\
&=\operatorname{trace}\Big(\big(F_\omega^*M_iP_{12,\omega}\Delta_{M_{k,i}}+M_iP_{12,\omega}\Delta_{M_{k,i}}F_{k,\omega}\big)P_{12}^T\Big)\nonumber\\
&=\operatorname{trace}\big(P_{12}^TF_\omega^*M_iP_{12,\omega}\Delta_{M_{k,i}}+F_{k,\omega}P_{12}^TM_iP_{12,\omega}\Delta_{M_{k,i}}\big)\nonumber\\
&=\operatorname{trace}\big((P_{12,\omega})^*M_iF_\omega P_{12}(\Delta_{M_{k,i}})^T+(P_{12,\omega})^*M_iP_{12}F_{k,\omega}^*(\Delta_{M_{k,i}})^T\big)\nonumber\\
&=\operatorname{trace}\big((P_{12,\omega})^*M_iP_{12,\omega}(\Delta_{M_{k,i}})^T\big).\nonumber
\end{align}
Additionally, note that
\begin{align}
\operatorname{trace}(B_k^T\Delta_{Q_{k,\omega}}^{M_{k,i}}B_k)&=\operatorname{trace}\big(B_kB_k^T\Delta_{Q_{k,\omega}}^{M_{k,i}}\big)\nonumber\\
&=\operatorname{trace}\Big(\big(-A_kP_k-P_kA_k^T\big)\Delta_{Q_{k,\omega}}^{M_{k,i}}\Big)\nonumber\\
&=\operatorname{trace}\Big(\big(-A_k^T\Delta_{Q_{k,\omega}}^{M_{k,i}}-\Delta_{Q_{k,\omega}}^{M_{k,i}}A_k\big)P_k\Big)\nonumber\\
&=\operatorname{trace}\Big(\big(F_{k,\omega}^*\Delta_{M_{k,i}}P_{k,\omega}M_{k,i}+\Delta_{M_{k,i}}P_{k,\omega}M_{k,i}F_{k,\omega}\nonumber\\
&+F_{k,\omega}^*M_{k,i}P_{k,\omega}\Delta_{M_{k,i}}+M_{k,i}P_{k,\omega}\Delta_{M_{k,i}}F_{k,\omega}\big)P_k\Big)\nonumber\\
&=\operatorname{trace}\big(F_{k,\omega}^*\Delta_{M_{k,i}}P_{k,\omega}M_{k,i}P_k+\Delta_{M_{k,i}}P_{k,\omega}M_{k,i}F_{k,\omega}P_k\nonumber\\
&+F_{k,\omega}^*M_{k,i}P_{k,\omega}\Delta_{M_{k,i}}P_k+M_{k,i}P_{k,\omega}\Delta_{M_{k,i}}F_{k,\omega}P_k\big)\nonumber\\
&=\operatorname{trace}\big(2P_{k,\omega}M_{k,i}P_{k,\omega}\big).\nonumber
\end{align}
Thus, $\Delta_{J}^{M_{k,i}}$ becomes
\begin{align}
\Delta_{J}^{M_{k,i}}=2\operatorname{trace}\big((-P_{12,\omega}^*M_iP_{12,\omega}+P_{k,\omega}M_{k,i}P_{k,\omega})(\Delta_{M_{k,i}})^T\big).\nonumber
\end{align}
Hence, the gradient of $J$ with respect to $M_{k,i}$ is given by
\begin{align}
\nabla_{J}^{M_{k,i}}=2(-P_{12,\omega}^*M_iP_{12,\omega}+P_{k,\omega}M_{k,i}P_{k,\omega}).\nonumber
\end{align}
Therefore,
\begin{align}
-P_{12,\omega}^*M_iP_{12,\omega}+P_{k,\omega}M_{k,i}P_{k,\omega}=0\nonumber
\end{align}is a necessary condition for the local optimum of $||G-G_k||_{\mathcal{H}_{2,\omega}}^2$.

By introducing a small first-order perturbation $\Delta_{B_k}$ to $B_k$, $J$ is perturbed to $J+\Delta_{J}^{B_k}$. Consequently, $P_{12,\omega}$, $P_{k,\omega}$, $Q_{12,\omega}$ and $Q_{k,\omega}$ are perturbed to $P_{12,\omega}+\Delta_{P_{12,\omega}}^{B_k}$, $P_{k,\omega}+\Delta_{P_{k,\omega}}^{B_k}$, $Q_{12,\omega}+\Delta_{Q_{12,\omega}}^{B_k}$, and $Q_{k,\omega}+\Delta_{Q_{k,\omega}}^{B_k}$, respectively. As a result, the first-order terms of $\Delta_{J}^{B_k}$ are given by
\begin{align}
\Delta_{J}^{B_k}&=\operatorname{trace}\big(-2Q_{12,\omega}^*B(\Delta_{B_k})^T+2Q_{k,\omega}B_k(\Delta_{B_k})^T\nonumber\\
&\hspace{3cm}-2BB_k(\Delta_{Q_{12,\omega}}^{B_k})^*+B_kB_k^T\Delta_{Q_{k,\omega}}^{B_k}\big).\nonumber
\end{align}
It follows from (\ref{eq:24})-(\ref{eq:31}) that $\Delta_{P_{12,\omega}}^{B_k}$, $\Delta_{P_{k,\omega}}^{B_k}$, $\Delta_{Q_{12,\omega}}^{B_k}$, and $\Delta_{Q_{k,\omega}}^{B_k}$ satisfy the following equations:
\begin{align}
&A\Delta_{P_{12,\omega}}^{B_k}+\Delta_{P_{12,\omega}}^{B_k}A_k^T+F_\omega B(\Delta_{B_k})^T+B(\Delta_{B_k})^TF_{k,\omega}^*=0,\nonumber\\
&A_k\Delta_{P_{k,\omega}}^{B_k}+\Delta_{P_{k,\omega}}^{B_k}A_k^T+F_{k,\omega}\Delta_{B_k}B_k^T+\Delta_{B_k}B_k^TF_{k,\omega}^*\nonumber\\
&\hspace*{3cm}+F_{k,\omega}B_k(\Delta_{B_k})^T+B_k(\Delta_{B_k})^TF_{k,\omega}^*=0\nonumber\\
&A^T\Delta_{Q_{12,\omega}}^{B_k}+\Delta_{Q_{12,\omega}}^{B_k}A_k+F_\omega^*\sum_{i=1}^{p}M_i\Delta_{P_{12,\omega}}^{B_k}M_{k,i}+\sum_{i=1}^{p}M_i\Delta_{P_{12,\omega}}^{B_k}M_{k,i}F_{k,\omega}=0\nonumber\\
&A_k^T\Delta_{Q_{k,\omega}}^{B_k}+\Delta_{Q_{k,\omega}}^{B_k}A_k+F_{k,\omega}^*\sum_{i=1}^{p}M_{k,i}\Delta_{P_{k,\omega}}^{B_k}M_{k,i}+\sum_{i=1}^{p}M_{k,i}\Delta_{P_{k,\omega}}^{B_k}M_{k,i}F_{k,\omega}=0.\nonumber
\end{align}
Note that
\begin{align}
&\operatorname{trace}\big(BB_k^T(\Delta_{Q_{12,\omega}}^{B_k})^*\big)\nonumber\\
&=\operatorname{trace}\Big(\big(-AP_{12}-P_{12}A_k^T\big)(\Delta_{Q_{12,\omega}}^{B_k})^*\Big)\nonumber\\
&=\operatorname{trace}\Big(\big(-A^T\Delta_{Q_{12,\omega}}^{B_k}-\Delta_{Q_{12,\omega}}^{B_k}A_k\big)P_{12}^T\Big)\nonumber\\
&=\operatorname{trace}\Big(\big(F_\omega^*\sum_{i=1}^{p}M_i\Delta_{P_{12,\omega}}^{B_k}M_{k,i}+\sum_{i=1}^{p}M_i\Delta_{P_{12,\omega}}^{B_k}M_{k,i}F_{k,\omega}\big)P_{12}^T\Big)\nonumber\\
&=\operatorname{trace}\Big(\sum_{i=1}^{p}M_{k,i}P_{12}^TF_\omega^*M_i\Delta_{P_{12,\omega}}^{B_k}+\sum_{i=1}^{p}M_{k,i}F_{k,\omega}P_{12}^TM_i\Delta_{P_{12,\omega}^{B_k}}\Big)\nonumber\\
&=\operatorname{trace}\Big(\sum_{i=1}^{p}M_iP_{12,\omega}M_{k,i}(\Delta_{P_{12,\omega}}^{B_k})^*\Big).\nonumber
\end{align}
Furthermore,
\begin{align}
&\operatorname{trace}\big(B_kB_k^T\Delta_{Q_{k,\omega}}^{B_k}\big)\nonumber\\
&=\operatorname{trace}\Big(\big(-A_kP_k-P_kA_k^T\big)\Delta_{Q_{k,\omega}}^{B_k}\Big)\nonumber\\
&=\operatorname{trace}\Big(\big(-A_k^T\Delta_{Q_{k,\omega}}^{B_k}-\Delta_{Q_{k,\omega}}^{B_k}A_k\big)P_k\Big)\nonumber\\
&=\operatorname{trace}\Big(\big(\sum_{i=1}^{p}F_{k,\omega}^*M_{k,i}\Delta_{P_{k,\omega}}^{B_k}M_{k,i}+\sum_{i=1}^{p}M_{k,i}\Delta_{P_{k,\omega}}^{B_k}M_{k,i}F_{k,\omega}\big)P_k\Big)\nonumber\\
&=\operatorname{trace}\Big(\sum_{i=1}^{p}M_{k,i}P_kF_{k,\omega}^*M_{k,i}\Delta_{P_{k,\omega}}^{B_k}+\sum_{i=1}^{p}M_{k,i}F_{k,\omega}P_kM_{k,i}\Delta_{P_{k,\omega}}^{B_k}\Big)\nonumber\\
&=\operatorname{trace}\Big(\sum_{i=1}^{p}M_{k,i}P_{k,\omega}M_{k,i}\Delta_{P_{k,\omega}}^{B_k}\Big).\nonumber
\end{align}
Thus,
\begin{align}
\Delta_{J}^{B_k}&=\operatorname{trace}\Big(-2Q_{12,\omega}^*B(\Delta_{B_k})^T+2Q_{k,\omega}B_k(\Delta_{B_k})^T\nonumber\\
&\hspace*{2cm}-2\sum_{i=1}^{p}M_iP_{12,\omega}M_{k,i}(\Delta_{P_{12,\omega}}^{B_k})^T+\sum_{i=1}^{p}M_{k,i}P_{k,\omega}M_{k,i}\Delta_{P_{k,\omega}}^{B_k}\Big).\nonumber
\end{align}
Additionally,
\begin{align}
&\operatorname{trace}\Big(\sum_{i=1}^{p}M_iP_{12,\omega}M_{k,i}(\Delta_{P_{12,\omega}}^{B_k})^*\Big)\nonumber\\
&=\operatorname{trace}\Big(\big(-A^T\bar{Z}_{12}-\bar{Z}_{12}A_k\big)(\Delta_{P_{12,\omega}}^{B_k})^*\Big)\nonumber\\
&=\operatorname{trace}\Big(\big(-A\Delta_{P_{12,\omega}}^{B_k}-\Delta_{P_{12,\omega}}^{B_k}A_k^T\big)\bar{Z}_{12}^*\Big)\nonumber\\
&=\operatorname{trace}\Big(\big(F_\omega B(\Delta_{B_k})^T+B(\Delta_{B_k})^TF_{k,\omega}^*\big)\bar{Z}_{12}^*\Big)\nonumber\\
&=\operatorname{trace}\big(\bar{Z}_{12}^*F_\omega B(\Delta_{B_k})^T+F_{k,\omega}^*\bar{Z}_{12}^*B(\Delta_{B_k})^T\big)\nonumber\\
&=\operatorname{trace}\big(Z_{12,\omega}^*B(\Delta_{B_k})^T\big).\nonumber
\end{align}
Similarly,
\begin{align}
&\operatorname{trace}\Big(\sum_{i=1}^{p}M_{k,i}P_{k,\omega}M_{k,i}\Delta_{P_{k,\omega}}^{B_k}\Big)\nonumber\\
&=\operatorname{trace}\Big(\big(-A_k^T\bar{Z}_k-\bar{Z}_kA_k\big)\Delta_{P_{k,\omega}}^{B_k}\Big)\nonumber\\
&=\operatorname{trace}\Big(\big(-A_k\Delta_{P_{k,\omega}}^{B_k}-\Delta_{P_{k,\omega}}^{B_k}A_k^T\big)\bar{Z}_k\Big)\nonumber\\
&=\operatorname{trace}\Big(2\bar{Z}_kF_{k,\omega}B_k(\Delta_{B_k})^T+2F_{k,\omega}^*\bar{Z}_kB_k(\Delta_{B_k})^T\Big)\nonumber\\
&=2\operatorname{trace}\big(Z_{k,\omega}^*B_k(\Delta_{B_k})^T\big).\nonumber
\end{align}
Thus,
\begin{align}
\Delta_{J}^{B_k}&=\operatorname{trace}\big(-2Q_{12,\omega}^*B(\Delta_{B_k})^T+2Q_{k,\omega}B_k(\Delta_{B_k})^T\nonumber\\
&\hspace*{2cm}-2Z_{12,\omega}^*B(\Delta_{B_k})^T+2Z_{k,\omega}B_k(\Delta_{B_k})^T\big).\nonumber
\end{align}
Therefore, the gradient of $J$ with respect to $B_k$ is
\begin{align}
\nabla_{J}^{B_k}=2(-Q_{12,\omega}^*B+Q_{k,\omega}B_k-Z_{12,\omega}^*B+Z_{k,\omega}B_k).\nonumber
\end{align}
Hence,
\begin{align}
&-Q_{12,\omega}^*B+Q_{k,\omega}B_k-Z_{12,\omega}^*B+Z_{k,\omega}B_k\nonumber\\
&=-(Y_{12,\omega}+2Z_{12,\omega})^*B+(Y_{k,\omega}+2Z_{k,\omega})B_k\nonumber\\
&=0\nonumber
\end{align} is a necessary condition for the local optimum of $||G-G_k||_{\mathcal{H}_{2,\omega}}^2$.

First, we reformat the cost function $J$ as follows:
\begin{align}
J&=\operatorname{trace}(-2B^T(Y_{12,\omega}+Z_{12,\omega})B_k+B_k^T(Y_{k,\omega}+Z_{k,\omega})B_k)\nonumber\\
&=\operatorname{trace}(-2B^TY_{12,\omega}B_k-2B^TZ_{12,\omega}B_k+B_k^TY_{k,\omega}B_k+B_k^TZ_{k,\omega}B_k).\nonumber
\end{align}
Since
\begin{align}
\operatorname{trace}(-2B^TY_{12,\omega}B_k+B_k^TY_{k,\omega}B_k)=\operatorname{trace}(2CP_{12,\omega}C_k^T+C_kP_{k,\omega}C_k^T),\nonumber
\end{align}
we can write
\begin{align}
J=\operatorname{trace}(2CP_{12,\omega}C_k^T+C_kP_{k,\omega}C_k^T-2B^TZ_{12,\omega}B_k+B_k^TZ_{k,\omega}B_k);\nonumber
\end{align}cf. \cite{zulfiqar2022adaptive}. By adding a small first-order perturbation $\Delta_{C_k}$ to $C_k$, the cost function $J$ is perturbed to $J+\Delta_{J}^{C_k}$. The first-order terms of $\Delta_{J}^{C_k}$ are given by:
\begin{align}
\Delta_{J}^{C_k}=\operatorname{trace}(2CP_{12,\omega}(\Delta_{C_k})^T+2C_kP_{k,\omega}(\Delta_{C_k})^T)\nonumber
\end{align}
Thus, the gradient of $J$ with respect to $C_k$ is:
\begin{align}
\nabla_{J}^{C_k}=2(CP_{12,\omega}+C_kP_{k,\omega}).\nonumber
\end{align}
Therefore, the necessary condition for the local optimum of $||G-G_k||_{\mathcal{H}_{2,\omega}}^2$ is:
\begin{align}
CP_{12,\omega}+C_kP_{k,\omega}=0.\nonumber
\end{align}This concludes the proof.

\end{document}